\documentclass[journal]{IEEEtran}
\usepackage{ifpdf}
\usepackage{latexsym,amsfonts,amssymb,amsmath,graphicx,epsf,cite,bbm,float}
\usepackage{ifpdf,flexisym}
\usepackage{epstopdf}
\usepackage{caption}
\usepackage{subcaption}
\usepackage{comment}
\usepackage{amsthm}
\usepackage{mathtools}
\usepackage{bbm}
\usepackage{color}
\usepackage{lettrine}
\usepackage{xfrac}

\newtheorem{theorem}{Theorem}
\newtheorem{lemma}{Lemma}

\newtheorem{definition}{Definition}

\begin{document}
\title{A Covert Queueing Channel in FCFS Schedulers}

\author
{AmirEmad Ghassami, \textit{Student Member, IEEE}, and Negar Kiyavash, \textit{Senior Member, IEEE}\thanks{AmirEmad Ghassami is with the Department of Electrical and Computer Engineering and Coordinated Science Lab, University of Illinois at Urbana-Champaign, Urbana, IL, email: \texttt{ghassam2@illinois.edu}. Negar Kiyavash is with the Department of Electrical and Computer Engineering, Department of Industrial and Enterprise Systems Engineering and Coordinated Science Lab, University of Illinois at Urbana-Champaign, Urbana, IL, email: \texttt{kiyavash@illinois.edu}.
The results in this paper were presented in part at the 2015 IEEE International Symposium on Information Theory, Hong Kong, June 14 - June 19, 2015 \cite{ghassami2015capacity}.}}

\maketitle

\begin{abstract}

We study covert queueing channels (CQCs), which are a kind of covert timing channel that may be exploited in shared queues across supposedly isolated users.
In our system model, a user sends messages to another user via his pattern of access to the shared resource, which serves the users according to a first come first served (FCFS) policy.
One example of such a channel is the cross-virtual network covert channel in data center networks, resulting from the queueing effects of the shared resource.
First, we study a system comprising a transmitter and a receiver that share a deterministic and work-conserving FCFS scheduler, and we compute the capacity of this channel.
We also consider the effect of the presence of other users on the information transmission rate of this channel.
The achievable information transmission rates obtained in this study demonstrate the possibility of significant information leakage and great privacy threats brought by CQCs in FCFS schedulers.


\end{abstract}


\section{Introduction}
\label{introduction}

The existence of side and covert channels due to the fragility of isolation mechanisms is an important privacy and security threat in computer networks. Such channels may be created across users, which were supposed to be isolated, resulting in information leakage. By definition, a covert channel is a hidden communication channel, which is not intended to exist in the system and is created furtively by users\cite{gligor1993covert}.
Covert channels may be exploited by a trusted user, or possibly a malware inside a system with access to secret information to leak it to a distrusted user.
On the other hand, in a side channel a malicious user attempts to learn private information by observing information not intended for him.
In this scenario, there is no collaboration between the source of information and the recipient \cite{kadloor2010low}.

Given that a lot of sensitive organizations such as CIA, and US Navy and Air Force, are abandoning in-house infrastructure and migrating to clouds, privacy has emerged as a serious risk for clouds. In the context of covert channels, a disgruntled employee, a leaker, or a malware, can easily sneak out extremely confidential and private data without explicitly communicating with an external party (data exfiltration). The employee could pretend to be talking to another trusted entity but sending out covert signals via shared queues to someone implicitly. This would be a severe threat to privacy leakage as it would be unbeknown to the tenant being targeted that their data is being exfiltrated and would bypass most defenses deployed at the host or network layer. Typical scenarios here could be stealing cryptographic keys, bank records, medical records, service records of military personal, names and locations of secret offices and networks of CIA/NSA, etc., i.e., small pieces of information that can have detrimental consequences if leaked.

Timing channels are one of the main types of covert/side channels, in which information is conveyed through timing of occurrence of events (e.g., inter-arrival times of packets).
A special case of timing channels is covert/side queueing channels, which can arise between users who share a packet scheduler in a network.
Packet schedulers serve packets from multiple streams, which are queued in a single queue. This causes dependencies between delays observed by users. Particularly, the delay that one user experiences depends on the amount of traffic generated by other streams, as well as his own traffic. Hence, a user can gain information about other users' traffic by observing delays of his own stream. This dependency between the streams can breach private information as well as create hidden communication channels between the users.

One example of a covert/side queueing channel is the cross-virtual network covert channel in data center networks and cloud environments.
As mentioned earlier, in recent years, migrating to commercial clouds and data centers is becoming increasingly popular among companies that deal with data.
The multi-tenant nature of cloud and sharing infrastructure between several users has made data protection and avoiding information leakage a serious challenge in such environments \cite{wsj2011}.
In data center networks, software-defined-networks are frequently used for load balancing \cite{mckeown2008openflow}. This generates logically isolated virtual networks and prevents direct data exchange. However, since packet flows belonging to different VNs inevitably share underlying network infrastructure (such as a router or a physical link), it is possible to transfer data across VNs through timing channels resulting from the queueing effects of the shared resource(s).

For data centers, the underlying premise is that resources should be multiplexed and oversubscribed as much as possible to maximize utilization and consequently profits. Hence, most of the resources are shared (e.g., cache, system bus, memory, NICs, switches, links, etc.). Under such a scenario, it can be seen that there will be numerous opportunities for establishing covert channels. Hence, a queue-sharing based channel is certainly not the only option. However, it is the most convenient/practical and perhaps the only one that continues to exist even if tenants resort to dedicated infrastructure (fully dedicated servers are provided by cloud providers at a higher cost, but network is always shared and no option is provided for dedicated networking infrastructure).
Two specific scenarios to argue this are as follows.
In the first case, there are two mutually untrusting tenants that are placed on separate servers. However, the first hop router and the outbound links from then onwards are shared. In such a scenarios, there are numerous shared queues which can be exploited to setup a covert channel.
Since the two tenants are deployed on separate servers, none of the other media can be exploited to create covert channels, and the only option is to use a network-based queue-sharing channel.
 In the second scenario, consider a load balancing server. Two separate tenants, each own a VM on this load balancing server (such servers have typically very high-end NICs and cloud vendors try to maximize sharing on these expensive NICs). The VMs on this server simply load balance traffic between other VMs of the tenant. Again, there are numerous network and server-based queues that can be exploited to setup a covert channel. 



In this paper, we study covert queueing channels (CQCs) in a shared deterministic and work-conserving first-come-first-served (FCFS) scheduler. We present an information-theoretic framework to describe and model the data transmission in this channel
and calculate its capacity.
First, we consider a two users setting depicted in Figure \ref{fig:2usersysmod}. In this model we have an encoder and a decoder user. Each user possesses a transmitter node and a receiver node. There is no direct communication channel between the users, but they share a packet scheduler. Hence, the delays observed by users are correlated. Therefore, the encoder user can encode a message in his traffic pattern and the decoder user can estimate the message by estimating the encoder's traffic pattern via the delays he experiences. Next, we extend the model to study
the effect of the presence of a third user on the information transmission rate.
The approach for analyzing the effect of the presence of the third user can be extended to calculate the capacity of the covert queueing channel serving any number of users.

The rest of the paper is organized as follows.
We review related works in Section \ref{relworks}.
In Section \ref{sysdesc}, we describe the system model.
The capacities of the introduced channel for the two and three user cases are calculated in Sections \ref{2user} and \ref{3user}, respectively.
Our concluding remarks are presented in Section \ref{conclusion}. 
\vspace{-3mm}
\section{Related Works}
\label{relworks}

The existing literature on covert/side timing channels has mainly concentrated on timing channels in which the receiver/adversary has direct access to the timing sequence produced by the transmitter/victim or a noisy version of it. However, in a covert/side queueing channel, the receiver/adversary does the inference based on the timing of his own packets which has been influenced by the original stream.

In a queuing side channel, where a malicious user, called an attacker, attempts to learn another user's private information, the main approach used by the attacker is traffic analysis. That is, the attacker tries to infer private information from the victim's traffic pattern. The attacker can have an estimation of the features of the other user's stream, such as packet size and timing by emitting frequent packets in his own sequence. Previous work shows that through traffic analysis, the attacker can obtain various private information including 
exact schedules of real-time systems \cite{chen2015schedule, certs2016shceduleak},
visited web sites \cite{liberatore2006inferring}, sent keystrokes \cite{song2001timing}, and even inferring spoken phrases in a voice-over-IP connection \cite{wright2010uncovering}.

In \cite{gong2012website}, Gong et al. proposed an attack  where a remote attacker learns about a legitimate user's browser activity by sampling the queue sizes in the downstream buffer of the user's DSL link.
The information leakage of a queueing side channel in an FCFS scheduler is analyzed in \cite{gong2011information}.
The analysis of more general work-conserving policies has been done in \cite{kadloor2012mitigating} and \cite{gong2014quantifying}.
The authors in \cite{gong2014quantifying} presented an analytical framework for modeling information leakage in queuing side channels  and quantify the leakage for several common scheduling policies.

Most of the work in covert timing channels is devoted to the case in which two users communicate by modulating the timings, and the receiver sees a noisy version of the transmitter's inputs \cite{anantharam1996bits, cabuk2004ip, murdoch2005embedding, llamas2005evaluation, kang1996network, gorantla2012characterizing, soltani2015covert, soltani2016covert, mukherjee2016covert}.
Also, there are many works devoted to the detection of such channels \cite{cabuk2004ip, berk2005detection, gianvecchio2007detecting}.
The setup of CQC is new in the field of covert channels and as far as the authors are aware, there are very few works on this setup \cite{ghassami2015capacity,tahir2015sneak, ghassami2017covert}. 
In \cite{ghassami2015capacity} we studied a system comprised of only a transmitter and a receiver and computed the capacity of the covert channel. We have extended the results of \cite{ghassami2015capacity} by studying the effect of the presence of other users on the information transmission rate of the CQC (Section \ref{3user}). Furthermore, the treatment of the proofs in the entire paper is now more rigorous and detailed. 
Specifically, we have provided a more in depth study of the function $\tilde{H}$ and its properties. This function indicates the highest amount of information transmittable for a given packet rate of the encoder user and a given inter-arrival time of packets of the decoder user, and
plays a fundamental role in the calculation of the capacity.

\section{System Description}
\label{sysdesc}

\begin{figure}[t]
\centering
\includegraphics[scale=0.395]{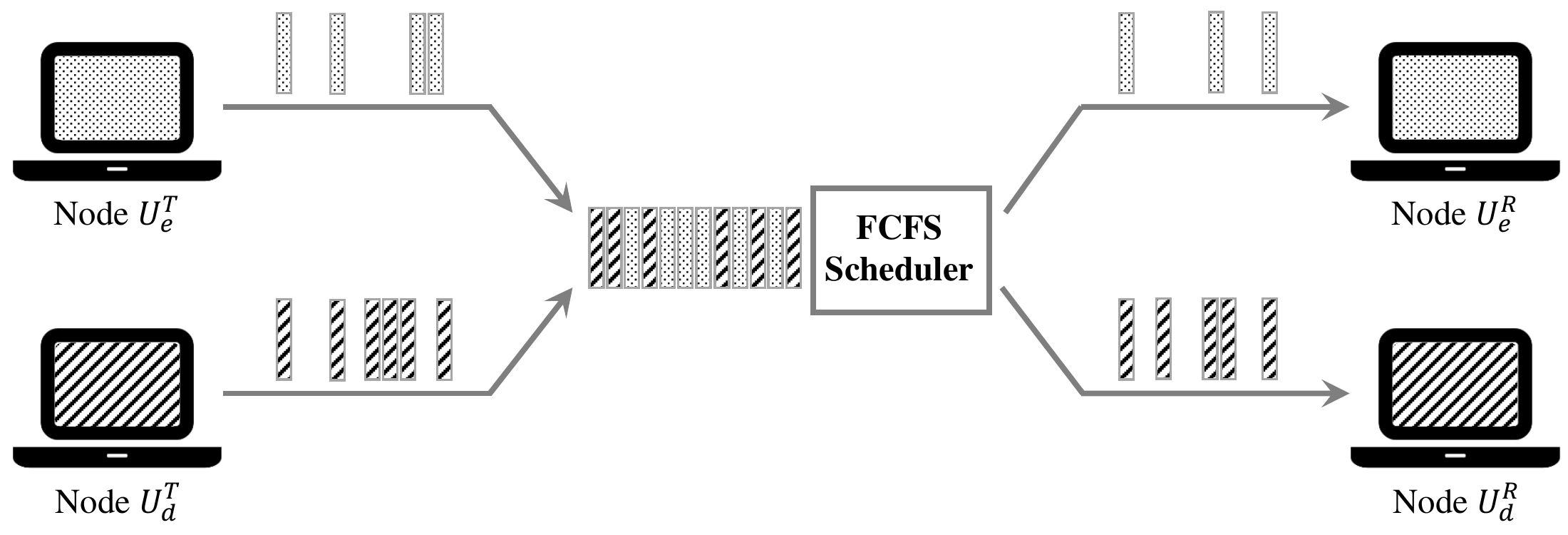}
\caption{Covert queueing channel in a system with 2 users.\vspace{-2mm}}
\label{fig:2usersysmod}
\end{figure}

Consider the architecture depicted in Figure \ref{fig:2usersysmod}. In this model, a scheduler serves packets from 2 users, $U_e$ and $U_d$.
Each user $U_i$, $i\in\{e,d\}$, is modeled by a transmitter and a receiver node, denoted by $U_i^T$ and $U_i^R$, respectively.
$U_i^R$ is the node which receives $U_i^T$'s packet stream. Note that $U_i^T$ and $U_i^R$ could correspond to the uplink and downlink of the same entity.
$U_e$ intends to send a message to $U_d$, but there is no direct channel between them. However, since $U_e^T$ and $U_d^T$'s packets share the same queue, $U_e^T$ can encode messages in the arrival times of its packets, which are passed onto  $U_d$ via queueing delays. Therefore, a timing channel is created between users via the delays experienced through the coupling of their traffic due to the shared scheduler.

To receive the messages from $U_e$, user $U_d$ sends a packet stream from the node $U_d^T$. He then uses the delays he experiences by receiving the packet stream at $U_d^R$ to decode the message.
Therefore, effectively, the nodes $U_e^T$ and $U_e^R$ are on the encoder side and the nodes $U_d^T$ and $U_d^R$ are on the decoder side of the channel of our interest.
Throughout the paper, we call $U_d$'s sent stream the \textit{probe stream}.

We consider an FCFS scheduler, which is commonly used in DSL routers. We assume this scheduler is deterministic and work-conserving. Time is discretized into slots, and the scheduler is capable of processing at most one packet per time slot. At each time slot, each user either issues one packet or remains idle. Furthermore, we assume that all packets are the same size.
Throughout the paper, we assume that the priorities of the users are known. Particularly, without loss of generality, we assume that $U_d$ has the highest priority among all users; i.e., in the case of simultaneous arrivals, $U_d$'s packet will be served first.

Figure \ref{fig:2userex} shows an example of the input and output streams of the system depicted in Figure \ref{fig:2usersysmod} with an FCFS scheduler. In this figure, the first stream is the arrival stream i.e., arrivals from both $U_e^T$ and $U_d^T$, depicted by dotted and diagonal patterns, respectively. The second stream is the output stream of user $U_e$ (received by $U_e^R$), and the third one is the output stream of user $U_d$ (received by $U_d^R$). In this example, we assume that one packet is buffered in the queue at time $A_i$, where a packet arrives from both $U_d^T$ and $U_e^T$.
If user $U_e$ had not sent the two packets (depicted by dotted pattern), the second packet of user $U_d$, which arrives at time $A_{i+1}$ could have departed one time slot earlier. Therefore, $U_d$ knows that $U_e$ has issued two packets.

\begin{figure}[t]
\centering
\includegraphics[scale=0.49]{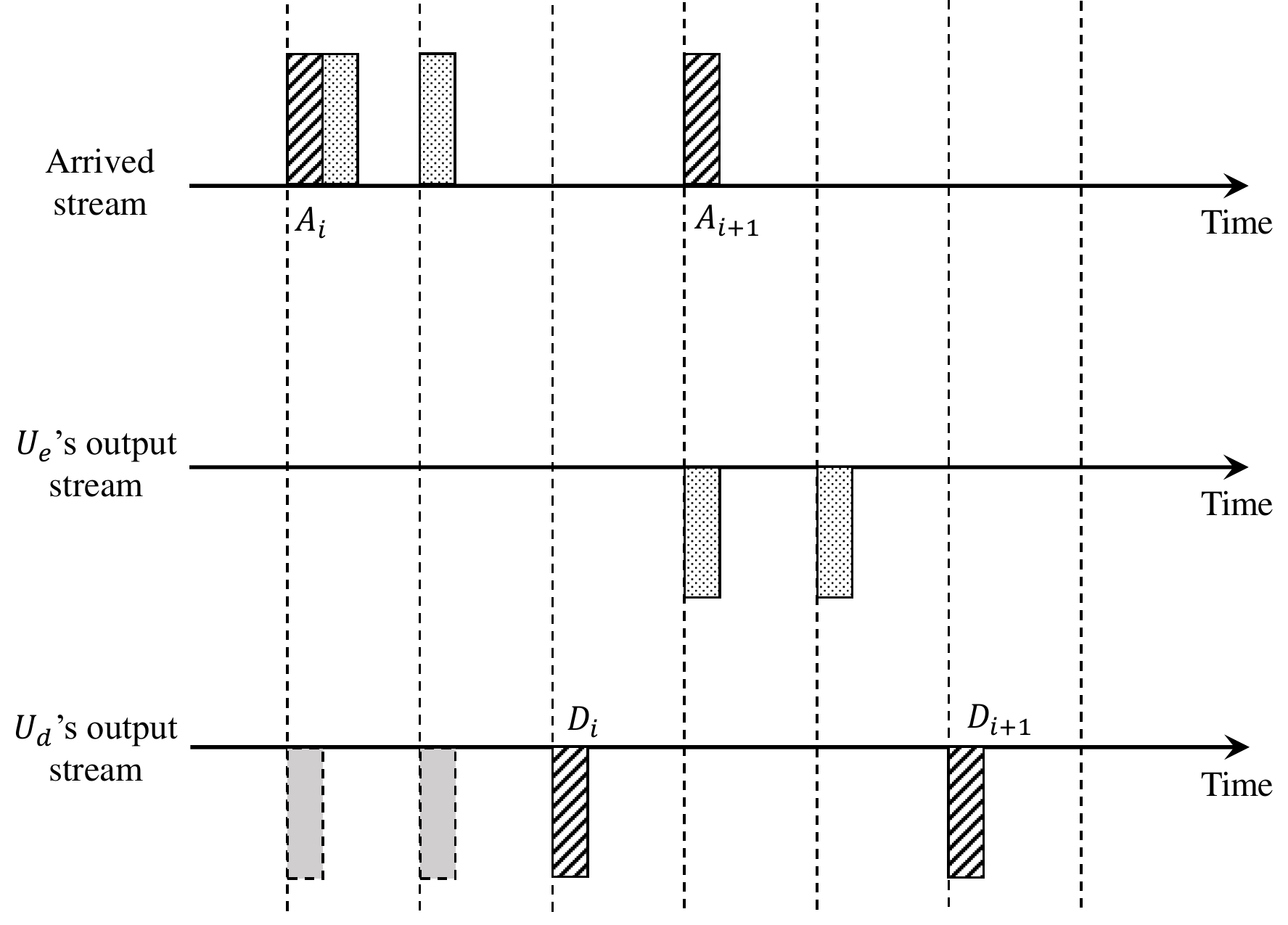}
\caption{An example of the input and output streams of the FCFS scheduler serving two users. Time slots are separated by dotted lines. Packets with dotted pattern belong to $U_e$ and packets with diagonal pattern belong to $U_d$. Gray packets are existing packets in the buffer, which could belong to either of the users. We assume that one packet is buffered in the queue at time $A_i$.\vspace{-0mm}}
\label{fig:2userex}
\end{figure}


As mentioned earlier, at each time slot, each user is allowed to either send one packet or none; hence, the input and output packet sequences of each user could be viewed as  a binary bitstream, where `1' and `0' indicates whether a packet was sent or not in the corresponding time slot.

Assume message $W$ drawn uniformly from the message set $\{1,2,...,M\}$ is transmitted by $U_e^T$, and $\hat{W}$ is $U_d$'s estimate of the sent message. Our performance metric is the average error probability, defined as follows.
\begin{equation*}
P_e\triangleq P(W\neq\hat{W})=\displaystyle\sum_{m=1}^M{\frac{1}{M}P(\hat{W}\neq m|W=m)}.
\end{equation*}
$U_e$ encodes each message into a binary sequence of length $n$, $\Delta^n$, to create the codebook, which is known at the decoder, $U_d$.

In order to send a message, $U_e^T$ emits a packet in the $i$-th time slot if $\Delta_i=1$ and remains idle otherwise, i.e.,
\begin{equation*}
\Delta_i = \left\{
\begin{array}{l l}
1 & \quad \Rightarrow U_e^T \text{ issues a}\\ &\quad \text{  packet in time slot $i$.}\\
0 & \quad \Rightarrow U_e^T \text{ remains}\\ &\quad \text{ idle in time slot $i$.}
\end{array} \right.
\end{equation*}

To decode this message, $U_d^T$ sends a binary length $n$ stream (the probe stream) to the scheduler during the same length $n$ time period. User $U_d$ will use this stream and the response stream received at node $U_d^R$ to decode the sent message.\\

We define the code, rate of the code, and the channel capacity similar to the definitions in \cite{anantharam1996bits}, \cite{cover2012elements}, and \cite{csiszar2011information}, as follows.

\begin{definition}
An $(n, M, \epsilon )$-code consists of a codebook of $M$ equiprobable binary codewords,
where messages take on average $n$ time slots to be received,
and the error probability satisfies $P_e\leq\epsilon$.
\end{definition}

\begin{definition}
The information transmission rate, $R$, of a code is the amount of conveyed information (logarithm of the codebook size) normalized by the average number of used time slots for the message to be received, i.e.,
$
R=(\log{M})/n.
$
\end{definition}
\noindent
Rate may be interpreted as the average amount of information conveyed in a single time slot.
\begin{definition}[Channel Capacity]
The Shannon capacity $C$, for a channel is the maximum achievable rate at which one can communicate through the channel when the average probability of error goes to zero. In other words, $C$ is the supremum of rates $R$, which satisfy the following property\cite{csiszar2011information}.
\begin{equation*}
\begin{aligned}
&\displaystyle\forall \delta>0,\hspace{2mm}\exists (n, M, \epsilon_n )\text{-code}\\
&\hspace{25mm}\text{s.t.}\hspace{5mm} \left\{
\begin{array}{l l}
\frac{\log{M}}{n}>R-\delta \\
\epsilon_n\rightarrow0
\end{array} \right.\text{ as }n\rightarrow\infty.
\end{aligned}
\label{eq:cap}
\end{equation*}
\end{definition}

The following notations will be used throughout the paper.
\begin{itemize}
\item
$r_i$: $U_i^T$'s packet rate.
\item
$A_i$: Arrival time of the $i$-th packet in the probe stream.
\item
$D_i$: Departure time of the $i$-th packet of the probe stream.\\
We assume $m$ packets are sent by $U_d$ during $n$ time slots and we have: $\displaystyle r_d = \lim_{n\rightarrow\infty}{m}/{n}$.
\item
$X_i$: Number of $U_e$'s packets sent in the interval $[A_i,A_{i+1})$. Note that $X_i=\sum_{j=A_i}^{A_{i+1}-1}\Delta_j$.
\item
$T_i = A_{i+1}-A_i$: Inter-arrival time between $i$-th and $(i+1)$-th packet of the probe stream. We denote a realization of $T$ by $\tau$.
\item
$Y_i = D_{i+1}-D_i-1$.
\item
$\hat{X}_i$: Estimation of $ X_i$ by decoder.
\item
$\hat{W}$: Decoded message.
\end{itemize}

In an FCFS scheduler, $U_d$ can have an estimation of the number of the packets of other users between any of his own consecutive packets. The estimation of the number of packets in the interval $[A_i, A_{i+1})$ is accurate if the scheduler is deterministic and work-conserving and a sufficient number of packets is buffered in the queue at time $A_i$\footnote{If the service rate of the scheduler is equal to $1$, there should be at least $A_{i+1}-A_i-1$ packets buffered in the queue at time $A_i$. Therefore, user $U_d$ needs to know the queue length. This is feasible using the following formula.
\begin{equation*}
q(A_i)=D_i-A_i-1
\end{equation*}
where $q(A_i)$ denotes the queue length at the time that the $i$-th packet in the probe stream arrives at the queue. The extra 1 in the formula is the time needed for the $i$-th packet of the probe stream to be served.
Therefore, user $U_d$ should always be aware of the queue length and keep it sufficiently large by sending extra packets when needed.
}. In that case, the number of other users' packets arriving in the interval $[A_i, A_{i+1})$ could be simply calculated by $Y_i=D_{i+1}-D_i-1$. Note that $U_d$ cannot pinpoint the location of the sent packets; that is, if the inter-arrival time is $\tau$, $U_d$ can distinguish between $\tau+1$ different sets of bit streams sent during this time.
Therefore, we look at any probe stream sent during $n$ time slots as a combination of different inter-arrival times.

If the sum of the packet rates of the users used during sending a message of length $n$ is on average larger than 1, then the message will be arrived on average during more than $n$ time slots. Also, this will destabilize the input queue of the scheduler. For example, for a system with two users $U_d$ and $U_e$,
if $U_d^T$ sends packets in every time slot, then sending a packet by $U_e^T$ in any time slot would cause a delay in the serving of the next packet of $U_d^T$ and hence could be detected. Therefore, in each time slot, $U_e^T$ could simply idle to signal a bit `0' or send a packet to signal a bit `1', resulting in the information rate of ${1}/{1.5}$ bit per time slot in the case that bits are equiprobable. But, this scheme is  not feasible in practice as it would destabilize the queue and result in severe packet drops.

In our model, we do not have any restrictions on the number of time slots in which the data transmission is happening. Therefore, the arrival sequence can be arbitrarily long, and hence, queue stability is required.
In order to have queue stability, it suffices that the total packet arrival rate does not exceed the service rate, which for a deterministic and work-conserving scheduler is equal to 1 (see Appendix \ref{app:A} for the proof of stability which is based on a Lyapunov stability argument for the general case that the serving rate is assumed to be $0\le\rho\le1$ and arbitrary number of users is considered).
Specifically, for the case of two users we need
\begin{equation}
r_e+r_d<1.
\label{eq:constrain}
\end{equation}

On the other hand, if the sum of the packet rates of the users used during sending a message of length $n$ is on average less than 1,
the length of the input queue may become less than the value required for accurate estimation of the number of the packets of other users between any of $U_d$'s consecutive packets.
Consequently, this creates an extra source of error for $U_d$ in estimating $X_i$'s.
In other words, for any given scheme with sum of the packet rates less than 1, increasing $r_d$ increases the resolution available for user $U_d$ by removing the aforementioned extra source of error. Hence, $U_d$ can have a better estimation of the number of other users' sent packets. 
To avoid this extra source of error, we focus on the coding schemes where the sum of the rates is 1. 
Therefore, in the case of two users, in order to achieve the highest information rate, the operation point should tend to the line $r_e+r_d=1$.

\section{Two-user Case}
\label{2user}

In this section, using achievability and converse arguments, the capacity of the introduced system is calculated for a system with a deterministic and work-conserving FCFS scheduler serving packets from two users.

As depicted in Figure \ref{fig:2usersysmod}, user $U_e$ is attempting to send a message to $U_d$ through the covert queueing channel between them. Note that since we have considered service rate of 1 for the FCFS scheduler and users can agree on the packet stream sent by $U_d^T$ ahead of time,
the feedback $U_e^R$ is already available at the encoder.
Therefore, the following Markov chain holds
\begin{equation}
W\rightarrow X^m\rightarrow Y^m\rightarrow \hat{W}.
\label{eq:chain}
\end{equation}
Note that as mentioned earlier, if there is a sufficient number of packets buffered in the shared queue, $\hat{X}_i$ could be accurately estimated as $Y_i$.



The main result of this section is the following theorem, the proof of which is developed in the rest of the section.
\begin{theorem}
\label{2userthm}
The capacity of the timing channel in a shared FCFS scheduler with service rate 1 depicted in Figure \ref{fig:2usersysmod} is equal to $0.8114$ bits per time slot, which can be obtained by solving the following optimization problem.
\begin{equation}
\begin{aligned}
\displaystyle &C=\max_{\alpha , \gamma_1 , \gamma_2} \alpha\tilde{H}(\gamma_1,1)+(1-\alpha)\tilde{H}(\gamma_2,\frac{1}{2})\\
&s.t.~~~
\alpha(\gamma_1+1)+(1-\alpha)(\gamma_2+\frac{1}{2})=1,
\end{aligned}
\label{eq:capacit}
\end{equation}
where $0\leq\alpha\leq1$ and $0\leq\gamma_1, \gamma_2\leq\frac{1}{2}$ and the function $\tilde{H}:[0,1]\times\{\frac{1}{k}:k\in \mathbb{N}\}\mapsto[0,1]$ is defined as
\begin{equation}
\displaystyle\tilde{H}(\gamma,\frac{1}{k})=\frac{1}{k}\max_{\substack{X\in\{0,1,...,k\}\\ \mathbb{E}[X]=k\gamma}} H(X),\hspace{5mm}k\in\mathbb{N}, 0\leq\gamma\leq1.
\label{eq:h}
\end{equation}
\end{theorem}
\noindent
The function $\tilde{H}(\gamma,\frac{1}{k})$ indicates the ratio of  highest amount of information that a random variable with support $\{0,...,k\}$ and mean $k\gamma$ can contain, to the number of bits used for transmitting this information.
We first investigate some of the properties of the function $\tilde{H}$.

\begin{lemma}
\label{2userlem0}
Let $U_k\sim \textit{Unif}(\{0,1,...,k\})$.
The distribution which achieves the optimum value in (\ref{eq:h}) is the tilted version of $\textit{Unif}(\{0,1,...,k\})$ with parameter $\lambda$; that is
\begin{equation*}
\begin{aligned}
P_{X}(i)=\frac{e^{i\lambda}}{\sum_{i=0}^{k}e^{i\lambda}},~~~i\in\{0,1,...,k\},
\end{aligned}
\label{eq:dist}
\end{equation*}
 where $\lambda=(\psi'_{U_k})^{-1}(k\gamma)$,
where the function $\psi'_{U_k}(\cdot)$ is the derivative of the log-moment generating function of $U_k$.
\end{lemma}
\noindent
See Appendix \ref{app:B} for the proof of Lemma \ref{2userlem0}.

\begin{lemma}
\label{2userlem1}
The function
$\tilde{H}$
could be computed using the following expression.
\begin{equation}
\displaystyle\tilde{H}(\gamma,\frac{1}{k})=\frac{1}{k}[\log_2(k+1)-\psi^*_{U_k}(k\gamma)\log_2 e],
\label{eq:H_rf}
\end{equation}
where $U_k\sim \textit{Unif}(\{0,1,...,k\})$, and the function $\psi^*_{U_k}(\cdot)$ is the rate function given by the Legendre-Fenchel transform of the log-moment generating function $\psi_{U_k}(\cdot)$, 
\begin{equation}
\displaystyle\psi^*_{U_k}(\gamma)=\sup_{\lambda\in\mathbb{R}}\{\lambda\gamma-\psi_{U_k}(\lambda)\}.
\label{eq:ratefunc}
\end{equation}
\end{lemma}
\noindent
In order to prove this lemma, first we note that for any random variable $X$ defined over the set $\{0,1,...,k\}$,
\begin{equation*}
\begin{aligned}
H(X)&=\sum_{i=0}^kP_X(i)\log{\frac{1}{P_X(i)}}\\
&=\sum_{i=0}^kP_X(i)\log{(k+1)}-\sum_{i=0}^kP_X(i)\log{\frac{P_X(i)}{{\frac{1}{k+1}}}}\\
&=\log{(k+1)}-D(P_X||U_k),
\end{aligned}
\end{equation*}
where $D(P_X||U_k)$ denotes the KL-divergence between $P_X$ and $U_k$. Therefore, in order to maximize $H(X)$, we need to minimize $D(P_X||U_k)$. Using the following well-known fact\cite{polyanskiy2014lecture}, concludes the lemma.
\begin{equation}
\displaystyle\min_{\mathbb{E}[X]=k\gamma}D(P_X||U_k)=\psi^*_{U_k}(k\gamma)\log_2 e.
\label{eq:projection}
\end{equation}

\noindent
Figure \ref{fig:htilde} shows the function $\tilde{H}(\gamma,\frac{1}{k})$ for different values of $\gamma$ and $k\in\{1,2,3\}$.

\begin{figure}[t]
\centering
\begin{subfigure}{.5\textwidth}
  \centering
  \includegraphics[scale=0.33]{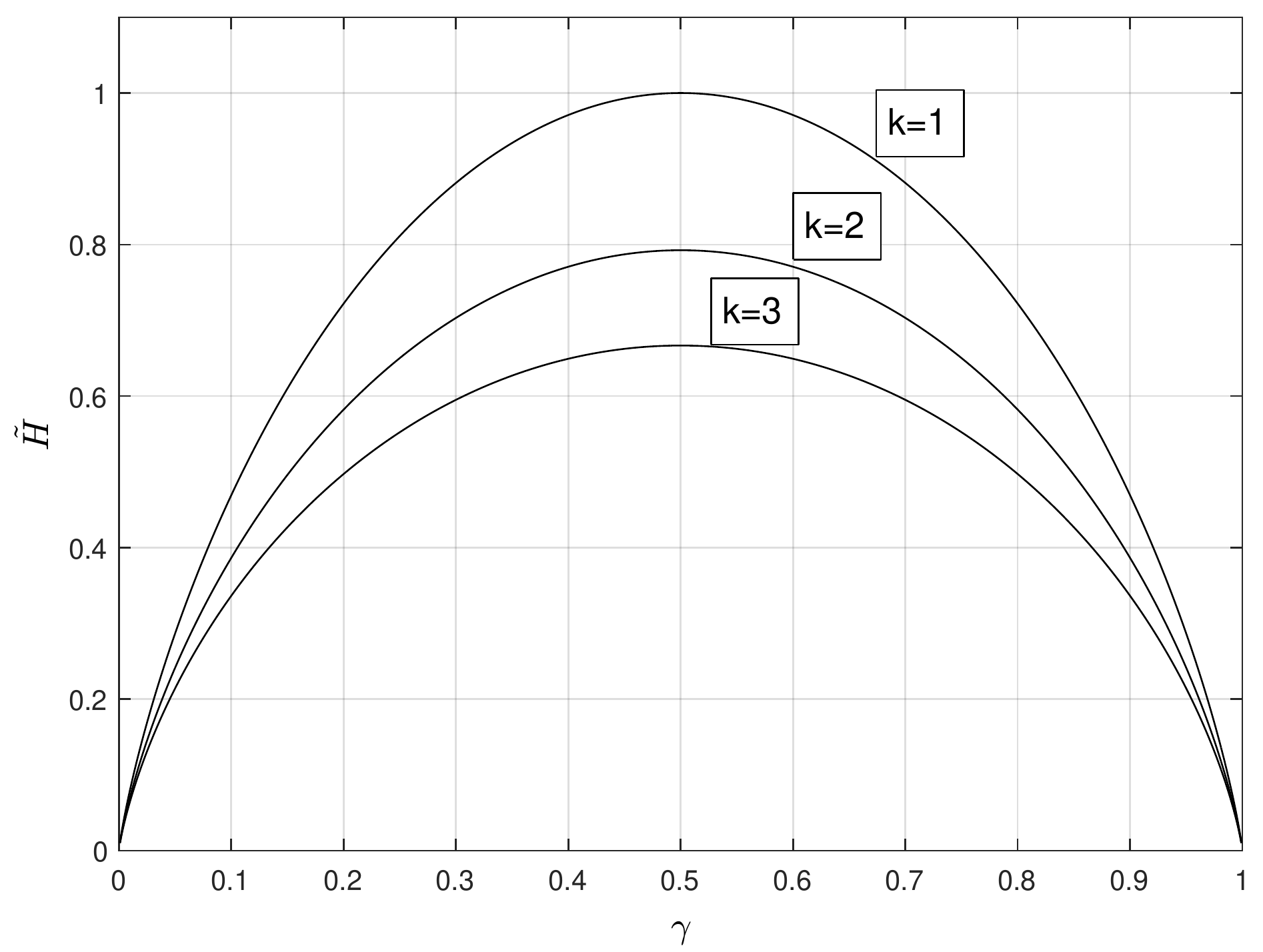}
\end{subfigure}\\
\begin{subfigure}{.5\textwidth}
  \centering
  \includegraphics[scale=0.33]{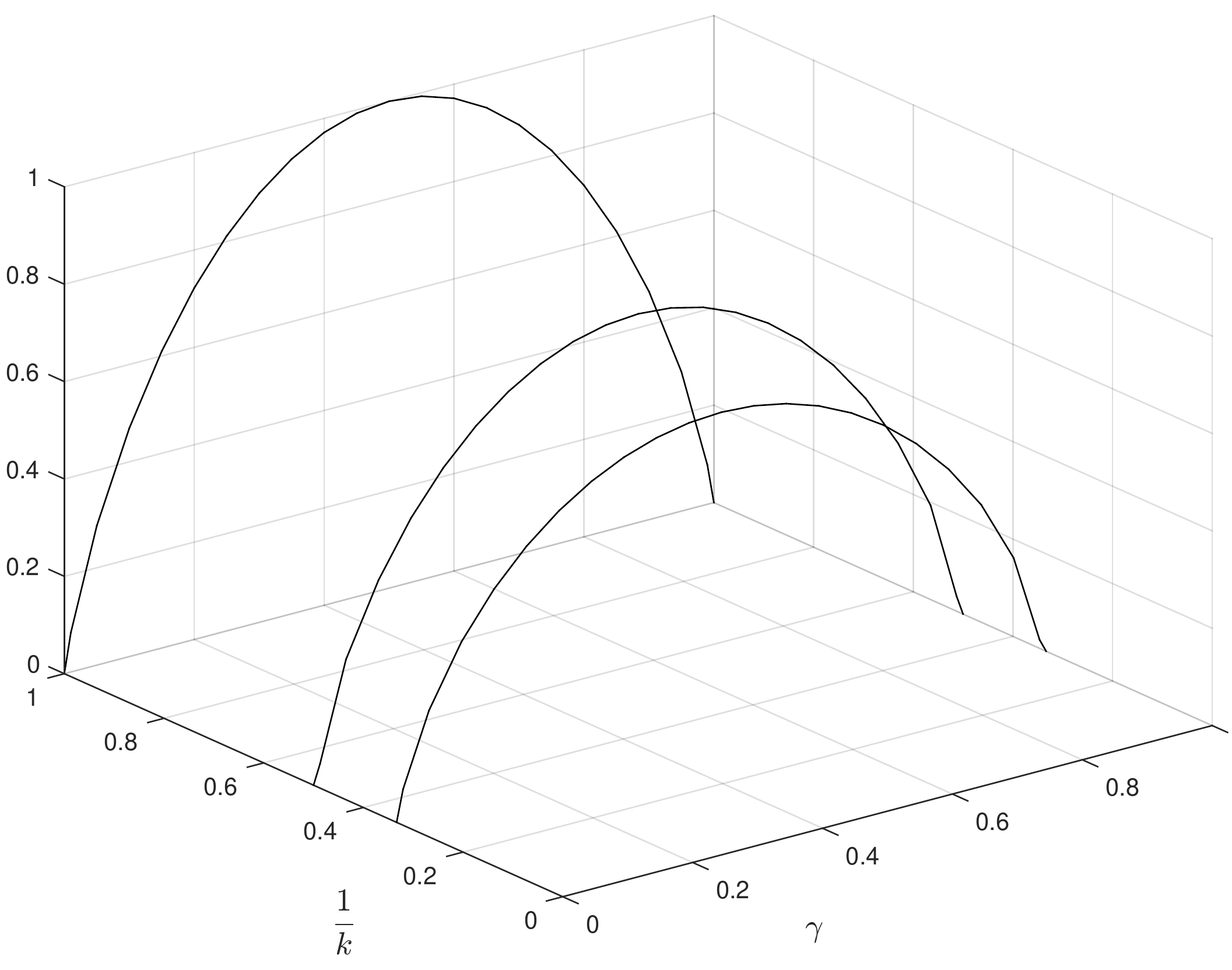}
\end{subfigure}
\caption{$\tilde{H}(\gamma,\frac{1}{k})$ for different values of $\gamma$ and $k\in\{1,2,3\}$.\vspace{-0mm}}
\label{fig:htilde}
\end{figure}


\begin{lemma}
\label{2userlem2}
The function $\tilde{H}(\cdot,\cdot)$ is concave in pair $\displaystyle(\gamma,\frac{1}{k})$ in the sense that for integers $k_1, k_2, k_3$, and for values $0\le\gamma_1,\gamma_2,\gamma_3\le1$, and for $\alpha\in[0,1]$, such that $\displaystyle\alpha(\gamma_1,\frac{1}{k_1})+(1-\alpha)(\gamma_3,\frac{1}{k_3})=(\gamma_2,\frac{1}{k_2})$, we have
\begin{equation}
\alpha\tilde{H}(\gamma_1,\frac{1}{k_1})+(1-\alpha)\tilde{H}(\gamma_3,\frac{1}{k_3})\le\tilde{H}(\gamma_2,\frac{1}{k_2}).
\label{eq:toolem}
\end{equation}
\end{lemma}
\noindent
See Appendix \ref{app:C} for the proof of Lemma \ref{2userlem2}.

Substituting (\ref{eq:H_rf}) in (\ref{eq:capacit}) and solving it, the capacity of the timing channel in the shared FCFS scheduler with service rate 1 depicted in Figure \ref{fig:2usersysmod} is equal to $0.8114$ bits per time slot, achieved by $\alpha = 0.177$, $\gamma_1 = 0.43$ and $\gamma_2 = 0.407$.

\begin{lemma}
\label{2userlem3}
For all $\gamma\in[0,1]$ and $k\in\mathbb{N}$, we have 
\[
\displaystyle\tilde{H}(\gamma, \frac{1}{k})=\tilde{H}(1-\gamma, \frac{1}{k}).
\]
\end{lemma}
\noindent
See Appendix \ref{app:D} for the proof of Lemma \ref{2userlem3}.

In the following, the proof of Theorem \ref{2userthm} is given. The proof is based on converse and achievability arguments.
\subsection{Converse}
In the converse side, the ultimate goal is to prove that
\begin{equation*}
\begin{aligned}
\displaystyle &C\le\max_{\alpha , \gamma_1 , \gamma_2} \alpha\tilde{H}(\gamma_1,1)+(1-\alpha)\tilde{H}(\gamma_2,\frac{1}{2})\\
&\text{s.t.}~~~
\alpha(\gamma_1+1)+(1-\alpha)(\gamma_2+\frac{1}{2})=1,
\end{aligned}
\label{eq:capacity}
\end{equation*}
where $0\leq\alpha\leq1$ and $0\leq\gamma_1, \gamma_2\leq\frac{1}{2}$.
We break the proof into two lemmas. First in Lemma \ref{2userlem4}, we find an upper bound on the information rate, which consists of a weighted summation of possible maximum information rates for different inter-arrival times of the packets in the probe stream, which satisfies the stability constraint. Then in Lemma \ref{2userlem5}, we upper bound the summation with one which only corresponds to inter-arrival times of 1 and 2.
\begin{lemma}
\label{2userlem4}
For the timing channel in a shared FCFS scheduler with service rate 1 depicted in Figure \ref{fig:2usersysmod}, any code consisting of a codebook of $M$ equiprobable binary codewords,
where messages take on average $n$ time slots to be received, satisfies
\begin{equation}
\begin{aligned}
\displaystyle\frac{1}{n}\log{M}\leq\sum_{\tau=1}^n[\pi_{\tau}\tilde{H}(\mu_{\tau}, \frac{1}{\tau})]+\epsilon_n,
\end{aligned}
\label{eq:generalsum}
\end{equation}
where
$\sum_{\tau=1}^n\pi_{\tau}(\mu_{\tau}+\frac{1}{\tau})=1$, and for all $\tau$, $0\le\mu_{\tau}\le\frac{1}{2}$.
In this expression, $\epsilon_n=\frac{1}{n}(H(P_e)+P_e\log_2{(M-1)})$, $\pi_{\tau}$ is the portion of time that user $U_d$ sends packets with inter-arrival time equal to $\tau$ in the probe stream, and $\mu_{\tau}$ is $U_e^T$'s average packet rate when the inter-arrival time is equal to $\tau$.
\end{lemma}
\begin{proof}
We first we note that
\begin{equation*}
\begin{aligned}
\displaystyle\frac{1}{n}\log{M}&\stackrel{(a)}{=}\frac{1}{n}H(W)\\
&\stackrel{(b)}{=}\frac{1}{n}H(W|\tau^m)\\
&=\frac{1}{n}I(W;\hat{W}|\tau^m)+\frac{1}{n}H(W|\hat{W}, \tau^m)\\
&\overset{(c)}{\le}\frac{1}{n}I(W;\hat{W}|\tau^m)+\epsilon_n\\
&\overset{(d)}{\le}\frac{1}{n}I(X^m;Y^m|\tau^m)+\epsilon_n,
\end{aligned}
\end{equation*}
where $(a)$ holds because $W$ is a uniform random variable over the set of messages $\{1,...,M\}$, $(b)$ follows from the fact that the chosen message is independent of  the inter-arrival time of decoder's packets, $(c)$ follows from Fano's inequality with $\epsilon_n=\frac{1}{n}(H(P_e)+P_e\log_2{(M-1)})$, and $(d)$ follows from data processing inequality in Markov chain in (\ref{eq:chain}). Therefore,
\begin{equation*}
\begin{aligned}
\displaystyle\frac{1}{n}\log{M}&\leq\frac{1}{n}[H(X^m|\tau^m)-H(X^m|Y^m, \tau^m)]+\epsilon_n\\
&\leq\frac{1}{n}H(X^m|\tau^m)+\epsilon_n\\
&\leq\frac{1}{n}\sum_{j=1}^mH(X_j|\tau^m)+\epsilon_n\\
&\leq\frac{1}{n}\sum_{j=1}^m\max_{P_{X_j|\tau^m}} H(X_j|\tau^m)+\epsilon_n.
\end{aligned}
\end{equation*}
In the maximization above, the mean of the distribution $P_{X_j|\tau^m}$ is $\mathbb{E}[X_j|\tau^m]$.
As mentioned in Section \ref{sysdesc}, in order to find the maximum information rate while having stability, we are interested in the asymptotic regime in which the operating point is converging to the line $r_e + r_d = 1$. Therefore, the information rate is upper bounded by having the set of means, $\{\mathbb{E}[X_1|\tau^m],\mathbb{E}[X_2|\tau^m],...,\mathbb{E}[X_m|\tau^m]\}$ satisfying the constraint $\frac{1}{n}\sum_{j=1}^m\mathbb{E}[X_j|\tau^m]+r_d= 1$.
Let $\xi_j=\frac{\mathbb{E}[X_j|\tau^m]}{\tau_j}$. Using (\ref{eq:h}), we have
\begin{equation}
\max_{\substack{P_{X_j|\tau^m}\\  \mathbb{E}[X_j|\tau^m]=\tau_j\xi_j}}H(X_j|\tau^m)=\tau_j\tilde{H}(\xi_j, \frac{1}{\tau_j}),
\label{eq:maxi1}
\end{equation}
where as mentioned in Lemma \ref{2userlem0}, the distribution for each $X_j$ which achieves the maximum value in (\ref{eq:maxi1}) is the tilted distribution of $U_{\tau_j}$ with parameter $\lambda$, such that
$\lambda=(\psi'_{U_{\tau_j}})^{-1}(\tau_j\xi_j)$.
Therefore, we will have:
\begin{equation*}
\begin{aligned}
\displaystyle\frac{1}{n}\log{M}\leq\frac{1}{n}\sum_{j=1}^m\tau_j\tilde{H}(\xi_j, \frac{1}{\tau_j})+\epsilon_n,
\end{aligned}
\end{equation*}
such that the set $\{ \xi_1,\xi_2,...,\xi_m\}$ satisfies the constraint  $\frac{1}{n}\sum_{j=1}^m\tau_j\xi_j+r_d= 1$.
The inter-arrival times take values in the set $\{1,2,...,n\}$.
Therefore, in the summation above we can fix the value of inter-arrival time on the value $\tau$ and count the number of times that $\tau_j$ has that value. Defining $m_\tau$ as the number of times that the inter-arrival time is equal to $\tau$ (note that $n = \sum_{\tau=1}^n \tau\cdot m_{\tau}$), we can break the summation above as follows.
\begin{equation*}
\begin{aligned}
\displaystyle\frac{1}{n}\log{M}&\leq\frac{1}{n}\sum_{\tau=1}^n[\sum_{k=1}^{m_{\tau}}\tau\tilde{H}(\mu_{\tau,k}, \frac{1}{\tau})]+\epsilon_n\\
&=\frac{1}{n}\sum_{\tau=1}^n[\tau\sum_{k=1}^{m_{\tau}}\tilde{H}(\mu_{\tau,k}, \frac{1}{\tau})]+\epsilon_n\\
&=\frac{1}{n}\sum_{\tau=1}^n[\tau\cdot m_{\tau}\sum_{k=1}^{m_{\tau}}\frac{1}{m_{\tau}}\tilde{H}(\mu_{\tau,k}, \frac{1}{\tau})]+\epsilon_n,
\end{aligned}
\end{equation*}
where $\displaystyle\mu_{\tau,k}$ is equal to the $k$-th $\xi_j$ which has $\tau_j = \tau$.

By Lemma \ref{2userlem2}, the function $\tilde{H}(\cdot,\cdot)$ is a concave function of its first argument. Therefore, by Jensen's inequality,
\begin{equation}
\sum_{k=1}^{m_{\tau}}\frac{1}{m_{\tau}}\tilde{H}(\mu_{\tau,k}, \frac{1}{\tau})\le\tilde{H}(\mu_{\tau}, \frac{1}{\tau}),
\label{eq:conv1p}
\end{equation}
where $\mu_{\tau}=\frac{1}{m_{\tau}}\sum_{k=1}^{m_{\tau}}\mu_{\tau,k}$.
Using (\ref{eq:conv1p}) and the equation $n = \sum_{\tau=1}^n \tau\cdot m_{\tau}$, we have
\begin{equation}
\begin{aligned}
\displaystyle\frac{1}{n}\log{M}&\leq\sum_{\tau=1}^n[\frac{\tau\cdot m_{\tau}}{\sum_{\tau=1}^n \tau\cdot m_{\tau}}\tilde{H}(\mu_{\tau}, \frac{1}{\tau})]+\epsilon_n\\
&=\sum_{\tau=1}^n[\pi_{\tau}\tilde{H}(\mu_{\tau}, \frac{1}{\tau})]+\epsilon_n,
\end{aligned}
\label{eq:generalsuminproof}
\end{equation}
where $\pi_{\tau} = (\tau\cdot m_{\tau})/(\sum_{\tau=1}^n \tau\cdot m_{\tau})$.

The packet rates of the users could be written as follows.
\begin{equation*}
\begin{aligned}
r_e&=\frac{1}{n}\sum_{j=1}^m\tau_j\xi_j = \frac{1}{n}\sum_{\tau=1}^n\sum_{k=1}^{m_{\tau}}\tau\mu_{\tau,k}\\
&= \frac{1}{n}\sum_{\tau=1}^n\tau m_{\tau}\frac{1}{m_{\tau}}\sum_{k=1}^{m_{\tau}}\mu_{\tau,k}= \sum_{\tau=1}^n\frac{1}{n}\tau m_{\tau}\mu_{\tau}\\
&= \sum_{\tau=1}^n\frac{\tau\cdot m_{\tau}}{\sum_{\tau=1}^n \tau\cdot m_{\tau}}\mu_{\tau}=\sum_{\tau=1}^n\pi_{\tau}\mu_{\tau},
\end{aligned}
\end{equation*}
and
\begin{equation*}
r_d=\sum_{\tau=1}^n\pi_{\tau}\frac{1}{\tau}.
\end{equation*}
Therefore, the constraint could be written as follows.
\begin{equation}
\displaystyle\sum_{\tau=1}^n\pi_{\tau}(\mu_{\tau}+\frac{1}{\tau}) = 1.
\label{eq:constrain2}
\end{equation}
Suppose the set of pairs $\{(\mu_{\tau},\frac{1}{\tau})\}_{\tau=1}^n$ satisfies (\ref{eq:constrain2}) and maximizes the right hand side of (\ref{eq:generalsuminproof}). By Lemma \ref{2userlem3}, there exists another set of pairs $\{(\hat{\mu}_{\tau},\frac{1}{\tau})\}_{\tau=1}^n$ with $\hat{\mu}_{\tau}$ defined as
\begin{equation*}
\hat{\mu}_{\tau} = \left\{
\begin{array}{l l}
\mu_{\tau} & \quad \text{if }0\leq\mu_{\tau}\leq\frac{1}{2},\\
1-\mu_{\tau} & \quad \text{if }\frac{1}{2}\leq\mu_{\tau}\leq1,
\end{array} \right.
\end{equation*}
that gives the same value for the right-hand side of (\ref{eq:generalsuminproof}), but it has $\sum_{\tau=1}^n\pi_{\tau}(\hat{\mu}_{\tau}+\frac{1}{\tau})\leq\sum_{\tau=1}^n\pi_{\tau}(\mu_{\tau}+\frac{1}{\tau})$. Therefore, $U_d$ can increase his packet rate and increase the information rate using the values $\hat{\mu}_{\tau}$.
Hence, in the maximizing set, for all $\tau$, we have $0\le\mu_{\tau}\le\frac{1}{2}$.
Therefore, the optimal operating point will be on the line $r_e + r_d = 1$, with $0\le r_e \le \frac{1}{2}$ and $\frac{1}{2}\le r_d \le 1$.
\end{proof}

Applying Lemma \ref{2userlem2}, we can replace all pairs of form $(\mu_{\tau},\frac{1}{\tau})$, $\tau\ge2$, with a single pair of form $(\mu,\frac{1}{2})$.
\begin{lemma}
\label{2userlem5}
For any set of pairs $\mathcal{S}=\{(\mu_{\tau},\frac{1}{\tau}),\tau\in[n]\}$, where for all $\tau$, $0\le\mu_{\tau}\le\frac{1}{2}$, with weights $\{\pi_{\tau},\tau\in[n]\}$
with operating point on the line $0\le r_e \le \frac{1}{2}$ and $\frac{1}{2}\le r_d \le 1$, there exists $0\leq\alpha\leq1$ and $0\leq\gamma_1, \gamma_2\leq\frac{1}{2}$ such that
\begin{equation*}
\alpha(\gamma_1+1)+(1-\alpha)(\gamma_2+\frac{1}{2})=1,
\end{equation*}
and
\begin{equation*}
\sum_{\tau=1}^n[{\pi}_{\tau}\tilde{H}({\mu}_{\tau}, \frac{1}{\tau})]\leq\alpha\tilde{H}(\gamma_1,1)+(1-\alpha)\tilde{H}(\gamma_2,\frac{1}{2}).
\end{equation*}
\end{lemma}

\begin{proof}
For all $\tau\in\{3,...,n\}$, there exists $\beta_{\tau}\in[0,1]$, such that
\begin{equation*}
\beta_{\tau}(\mu_1,1)+(1-\beta_{\tau})(\mu_{\tau},\frac{1}{\tau})=(\mu_2^{\tau},\frac{1}{2}).
\end{equation*}
Clearly, the set $\{(\mu_1,1),(\mu_{2},\frac{1}{2}),(\mu_2^{3},\frac{1}{2}),...,(\mu_2^{n},\frac{1}{2})\}$ can also give the same operating point as $\mathcal{S}$ does. By Lemma \ref{2userlem2},
\begin{equation*}
\beta_{\tau}\tilde{H}(\mu_1,1)+(1-\beta_{\tau})\tilde{H}(\mu_{\tau},\frac{1}{\tau})\le\tilde{H}(\mu_2^{\tau},\frac{1}{2}),\hspace{2mm}\forall\tau\in\{3,...,n\}.
\end{equation*}
Therefore,
\begin{equation}
\label{eq:goodtech}
\begin{aligned}
&\sum_{\tau=1}^n{\pi}_{\tau}\tilde{H}({\mu}_{\tau}, \frac{1}{\tau})\\
&=\zeta_1\tilde{H}(\mu_1,1)+\zeta_2\tilde{H}(\mu_2,\frac{1}{2})+\sum_{\tau=3}^n\zeta_{\tau}(\beta_{\tau}\tilde{H}(\mu_1,1)\\
&~~~+(1-\beta_{\tau})\tilde{H}(\mu_{\tau},\frac{1}{\tau}))\\
&\le\zeta_1\tilde{H}(\mu_1,1)+\zeta_2\tilde{H}(\mu_2,\frac{1}{2})+\sum_{\tau=3}^n\zeta_{\tau}\tilde{H}(\mu_2^{\tau},\frac{1}{2})\\
&\le\zeta_1\tilde{H}(\mu_1,1)+(1-\zeta_1)\tilde{H}(\frac{\zeta_2\mu_2+\sum_{\tau=3}^n\zeta_{\tau}\mu_2^{\tau}}{1-\zeta_1},\frac{1}{2}),
\end{aligned}
\end{equation}
where $\pi_1=\zeta_1+\sum_{\tau=3}^n\zeta_{\tau}\beta_\tau$, $\pi_2=\zeta_2$ and $\pi_{\tau}=\zeta_{\tau}(1-\beta_{\tau})$ for $3\le\tau\le n$, and we have used Lemma \ref{2userlem2} again in the last inequality.
\end{proof}

From Lemmas \ref{2userlem4} and \ref{2userlem5} we have
\begin{equation*}
\begin{aligned}
\displaystyle \frac{1}{n}\log{M}&\leq\alpha\tilde{H}(\gamma_1,1)+(1-\alpha)\tilde{H}(\gamma_2,\frac{1}{2})+\epsilon_n\\
&\leq\max_{\alpha , \gamma_1 , \gamma_2} \alpha\tilde{H}(\gamma_1,1)+(1-\alpha)\tilde{H}(\gamma_2,\frac{1}{2})+\epsilon_n.
\end{aligned}
\end{equation*}
Letting $n\rightarrow\infty$, $\epsilon_n$ goes to zero and we have
\begin{equation*}
\begin{aligned}
\displaystyle &C\le\max_{\alpha , \gamma_1 , \gamma_2} \alpha\tilde{H}(\gamma_1,1)+(1-\alpha)\tilde{H}(\gamma_2,\frac{1}{2})\\
&s.t.~~~
\alpha(\gamma_1+1)+(1-\alpha)(\gamma_2+\frac{1}{2})=1,
\end{aligned}
\label{eq:capacity}
\end{equation*}
where $0\leq\alpha\leq1$ and $0\leq\gamma_1, \gamma_2\leq\frac{1}{2}$.
This completes the proof of the converse part.
\subsection{Achievability}
\label{ssec:ach}
The sequence of steps in our achievability scheme is as follows:

\noindent
\textbf{Generating the codebook:} 
The codebook $\mathcal{C}$ is generated by combining two codebooks $\mathcal{C}_1$ and $\mathcal{C}_2$, which are designed for the cases that the inter-arrival times in the probe stream are 1 and 2, respectively. 
\begin{itemize}
\item Set $\alpha = 0.177-\delta$, for a small and positive value of $\delta$.
\item To generate $\mathcal{C}_1$, fix a binary distribution $P_1$ such that $P_1(1)=0.43$ and $P_1(0)=0.57$. Generate a binary codebook $\mathcal{C}_1$ containing $2^{\alpha nR_1}$ sequences of length $\alpha n$ of i.i.d. entries according to $P_1$.
\item To generate $\mathcal{C}_2$, fix a ternary distribution $P_2$ over set of symbols $\{a_0, a_1, a_2\}$ such that $P_2(a_0)=0.43$, $P_2(a_1)=0.325$ and $P_2(a_2)=0.245$. Generate a ternary codebook $\mathcal{C}_2$ containing $2^{(1-\alpha)nR_2}$ sequences of length $\frac{1}{2}(1-\alpha)n$ of i.i.d. entries according to $P_2$. Substituting $a_0$ with $00$, $a_1$ with $10$, and $a_2$ with $11$, we will have $2^{(1-\alpha)nR_2}$ binary sequences of length $(1-\alpha)n$.
\item Combine $\mathcal{C}_1$ and $\mathcal{C}_2$ to get $\mathcal{C}$, such that $\mathcal{C}$ has $2^{n(\alpha R_1+(1-\alpha)R_2)}$ binary sequences of length $n$ where we concatenate $i$-th row of $\mathcal{C}_1$ with $j$-th row of $\mathcal{C}_2$ to make the $((i-1)(2^{(1-\alpha)nR_2})+j)$-th row of $\mathcal{C}$ (note that $2^{(1-\alpha)nR_2}$ is the number of rows in $\mathcal{C}_2$). Rows of $\mathcal{C}$ are our codewords.
\end{itemize}
In above, $n$ should be chosen such that $\alpha nR_1$, $\alpha n$, $(1-\alpha)nR_2$ and $\frac{1}{2}(1-\alpha)n$ are all integers.

\noindent
\textbf{Encoding:} 
To send message $m$, $U_e^T$ sends the corresponding row of $\mathcal{C}$, that is, he sends the corresponding part of $m$ from $\mathcal{C}_1$ in the first $\alpha n$ time slots and the corresponding part of $m$ from $\mathcal{C}_2$ in the rest of $(1-\alpha)n$ time slots.

\noindent
\textbf{Decoding:} Recall that $U_d^T$ sends a binary length $n$ stream to the scheduler during the same length $n$ time period. Here, $U_d^T$ sends the stream of all ones (one packet in each time slot) in the first $\alpha n$ time slots and sends bit stream of concatenated $10$'s for the rest of $(1-\alpha)n$ time slots.
Assuming the queue is not empty,\footnote{Since in our achievable scheme $U_d^T$'s packets are spaced by either one or two time slots, it is enough to have one packet buffered in the queue, where since we are working in the heavy traffic regime, it will not be a problem.} since there is no  noise in the system, the decoder can always learn the exact sequence sent by $U_e$.

Consequently, we will have:
\begin{equation*}
\begin{aligned}
C&\geq\displaystyle\frac{\log_2{2^{n(\alpha R_1+(1-\alpha)R_2)}}}{n}\\
&=\alpha R_1+(1-\alpha)R_2.
\end{aligned}
\end{equation*}
In infinite block-length regime, where $n\rightarrow\infty$, we can choose $R_1 = H(P_1)$, $R_2 = \frac{1}{2}H(P_2)$ and find codebooks $\mathcal{C}_1$ and $\mathcal{C}_2$ such that this scheme satisfies the rate constraint. Therefore,
\begin{equation*}
C\geq\alpha H(P_1)+(1-\alpha)\frac{1}{2}H(P_2).
\end{equation*}
Substituting the values in the expression above, and letting $\delta$ go to zero, we see that the rate $0.8114$ bits per time slot is achievable.

\section{Three-user Case}
\label{3user}

\begin{figure}[t]
\centering
\includegraphics[scale=0.395]{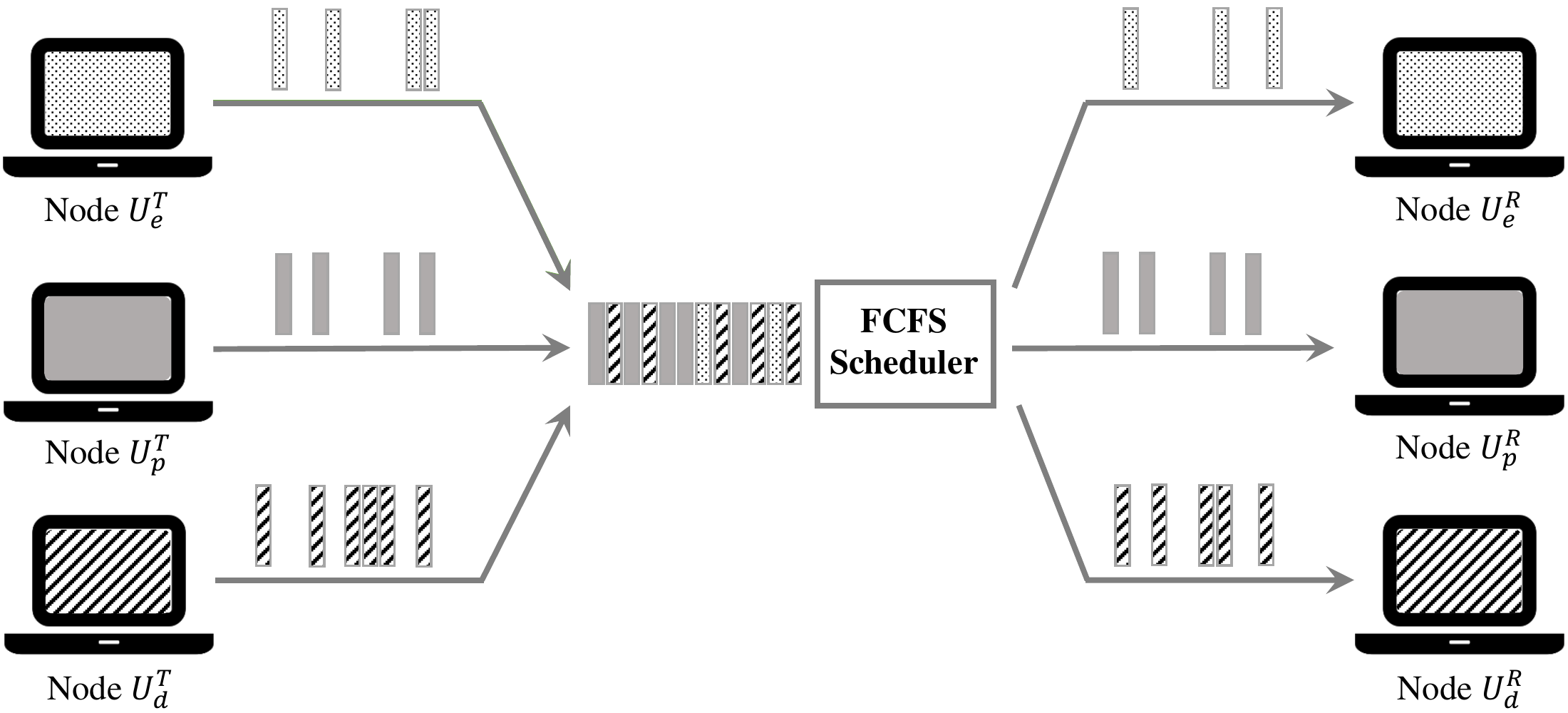}
\caption{Covert queueing channel in a system with three users.\vspace{-0mm}}
\label{fig:3usersysmod}
\end{figure}

As an extension to the basic problem, in this section we consider the case that a third user is also using the shared scheduler. We add user $U_p$ to our basic system model. This user has nodes $U_p^T$ and $U_p^R$ as his transmitter and receiver nodes, respectively (Figure \ref{fig:3usersysmod}). We assume that node $U_p^T$ sends packets according to a Bernoulli process with rate $r_p$ to the shared scheduler. The shared scheduler is again assumed to be FCFS with service rate 1 and we analyze the capacity for coding schemes satisfying queueing stability condition in the asymptotic regime where the operating point is converging to the line $r_e+r_p+r_d=1$. Also, in this section we consider the extra assumption that the inter-arrival time of the packets in the probe stream is upper bounded by the value $\tau_{max}$.

Assuming that a sufficient number of packets are buffered in the shared queue, user $U_d$ can still count the number of packets sent by the other two users between any of his own consecutive packets, yet he cannot distinguish between packets sent by user $U_e$ and the packets sent by user $U_p$. Hence, user $U_d$ has uncertainty in estimating the values of $X$. We model this uncertainty as a noise in receiving $X$. Suppose $U_d^T$ sends two packets with $\tau_i=2$. Each of the other users can possibly send at most 2 packets in the interval $[A_i,A_{i+1})$ and hence, $Y\in\{0,1,2,3,4\}$. Therefore, we have the channel shown in Figure \ref{fig:3userch} for this example.
\begin{figure}[t]
\centering
\includegraphics[scale=0.4]{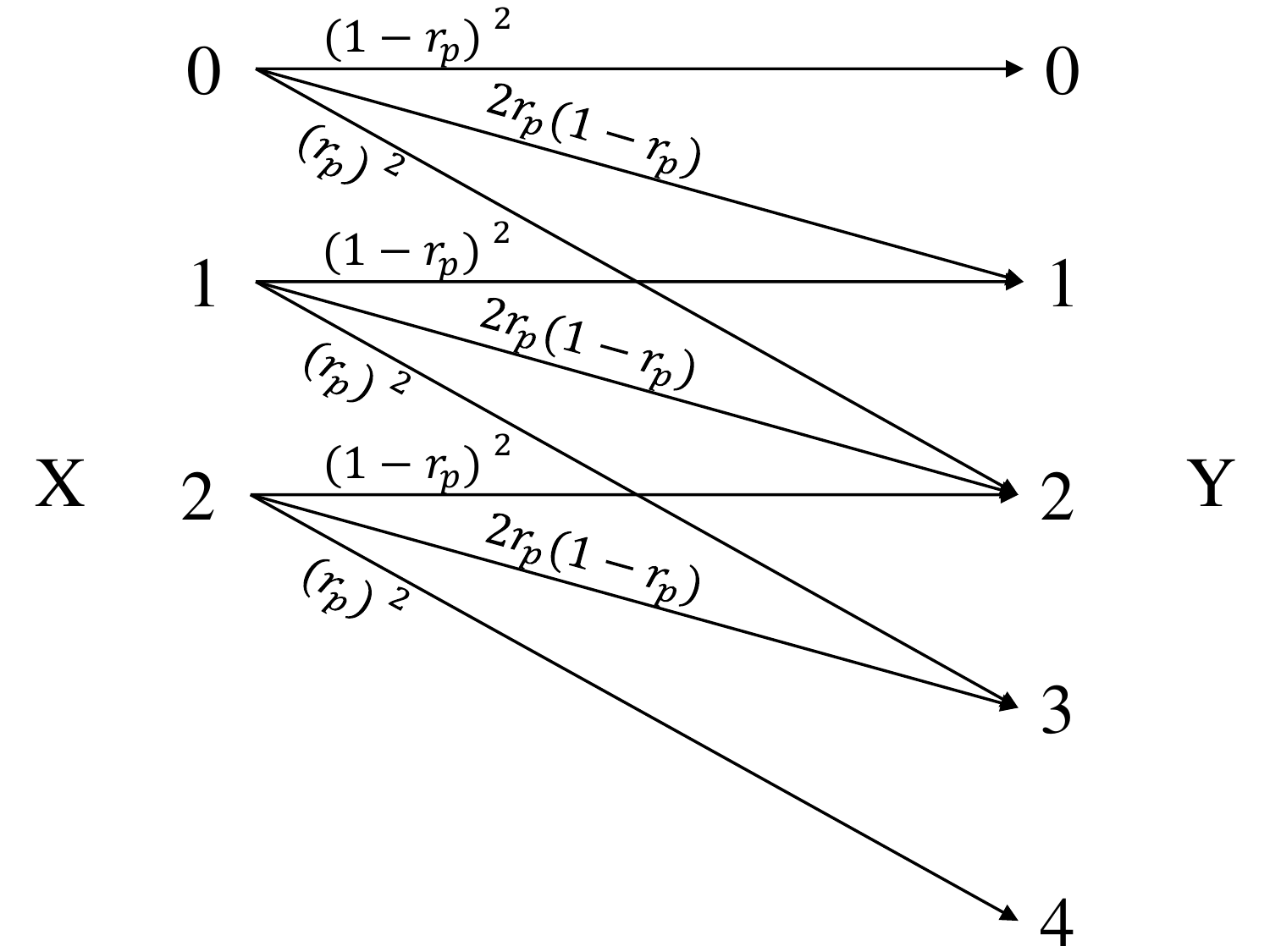}
\caption{Channel between the encoder and the decoder of the system for the case that the inter-arrival time of two packets of the probe stream is 2.\vspace{-0mm}}
\label{fig:3userch}
\end{figure}
In the general case, for the inter-arrival time $\tau$, given $X=x$, we have $Y\in\{x+0,...,x+\tau\}$ such that
\begin{equation}
\label{eq:yisbin}
Pr(Y=i+x|X=x)={\tau \choose i}(r_p)^i(1-r_p)^{\tau-i},\hspace{2mm}i\in\{0,...,\tau\},
\end{equation}
which is a binomial distribution $\textit{Bin}(\tau,r_p)$.
Therefore, the support of the random variable $Y$ is $\{0,1,...,2\tau\}$. For the mean of $Y$, we have
\begin{equation}
\mathbb{E}[Y|\tau]=\mathbb{E}[\mathbb{E}[Y|X,\tau]|\tau]=\mathbb{E}[X+\tau r_p|\tau]=\tau(r_e+r_p).
\end{equation}

Because of user $U_p$'s stream, the encoder is not aware of the stream received at node $U_e^R$ beforehand and this output can provide information to the encoder about $U_p^T$'s stream. The more packets node $U_e^T$ sends to the scheduler, the more information this stream contains about $U_p^T$'s stream. Using this information, the encoder can have an estimation of the output of the channel at the decoder's side and hence, it could be considered as a noisy feedback to the encoder. Figure \ref{fig:3userprob} shows the graphical model for random variables in our system.
\begin{figure}[t]
\centering
\includegraphics[scale=0.5]{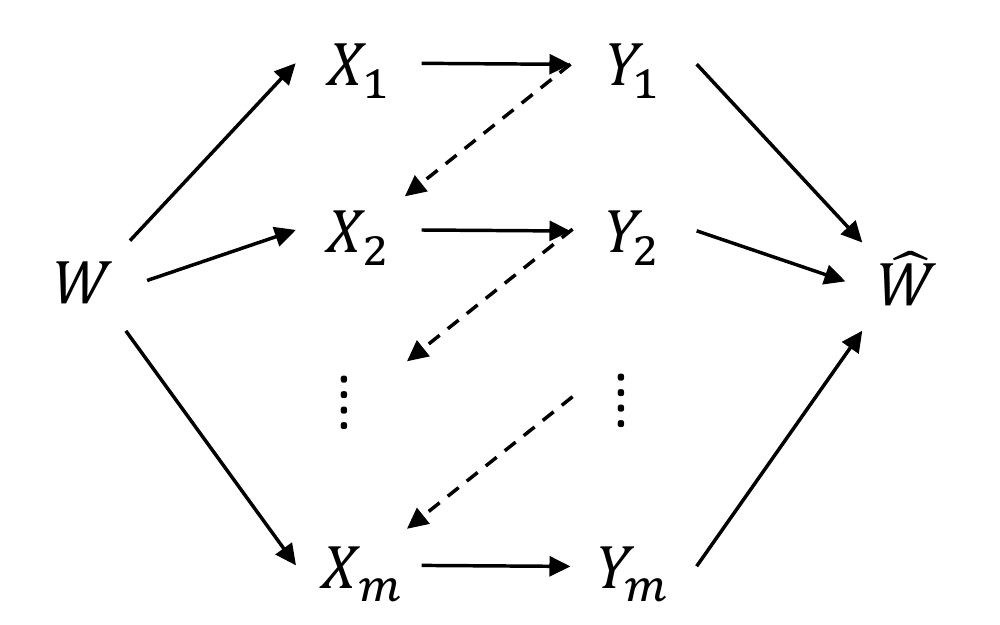}
\caption{The graphical model representing the statistical relation between $W$, $X^m$, $Y^m$ and $\hat{W}$.}
\label{fig:3userprob}
\end{figure}

The main result of this section is evaluation of the capacity of the introduced channel, presented in Theorem \ref{3userthm}. In the following, the subscript $r_p$ denotes that the calculation is done when the rate of $U_p$ is $r_p$.

\begin{theorem}
\label{3userthm}
If the rate of $U_p$ is $r_p$, the capacity of the timing channel in a shared FCFS scheduler with service rate 1 depicted in Figure \ref{fig:3usersysmod} is given by
\begin{equation}
\begin{aligned}
\displaystyle &C(r_p)=\max_{\alpha , \gamma_1 , \gamma_2,\tau} \alpha\tilde{I}_{r_p}(\gamma_1,\frac{1}{\tau})+(1-\alpha)\tilde{I}_{r_p}(\gamma_2,\frac{1}{\tau+1})\\
&s.t.~~~
\alpha(\gamma_1+\frac{1}{\tau})+(1-\alpha)(\gamma_2+\frac{1}{\tau+1})=1-r_p,
\end{aligned}
\label{eq:3capacit}
\end{equation}
where $0\leq\alpha\leq1$ and $0\leq\gamma_1, \gamma_2\le1$ and $1\le\tau\le\tau_{max}-1$. The function $\tilde{I}_{r_p}:[0,1]\times\{\frac{1}{k}:k\in \mathbb{N}\}\mapsto[0,1]$ is defined as
\begin{equation}
\displaystyle\tilde{I}_{r_p}(\gamma,\frac{1}{k})=\frac{1}{k}\max_{\substack{X\in\{0,1,...,k\}\\ \mathbb{E}[X]=k\gamma}}I_{r_p}(X;Y),\hspace{5mm}k\in\mathbb{N}, 0\leq\gamma\leq1.
\label{eq:i}
\end{equation}
\end{theorem}

The proof is based on converse and achievability arguments. Before giving the proof, we first investigate some of the properties of the function $\tilde{I}$.

\begin{figure*}[t]
\centering
\begin{subfigure}{.33\textwidth}
  \centering
  \includegraphics[scale=0.28]{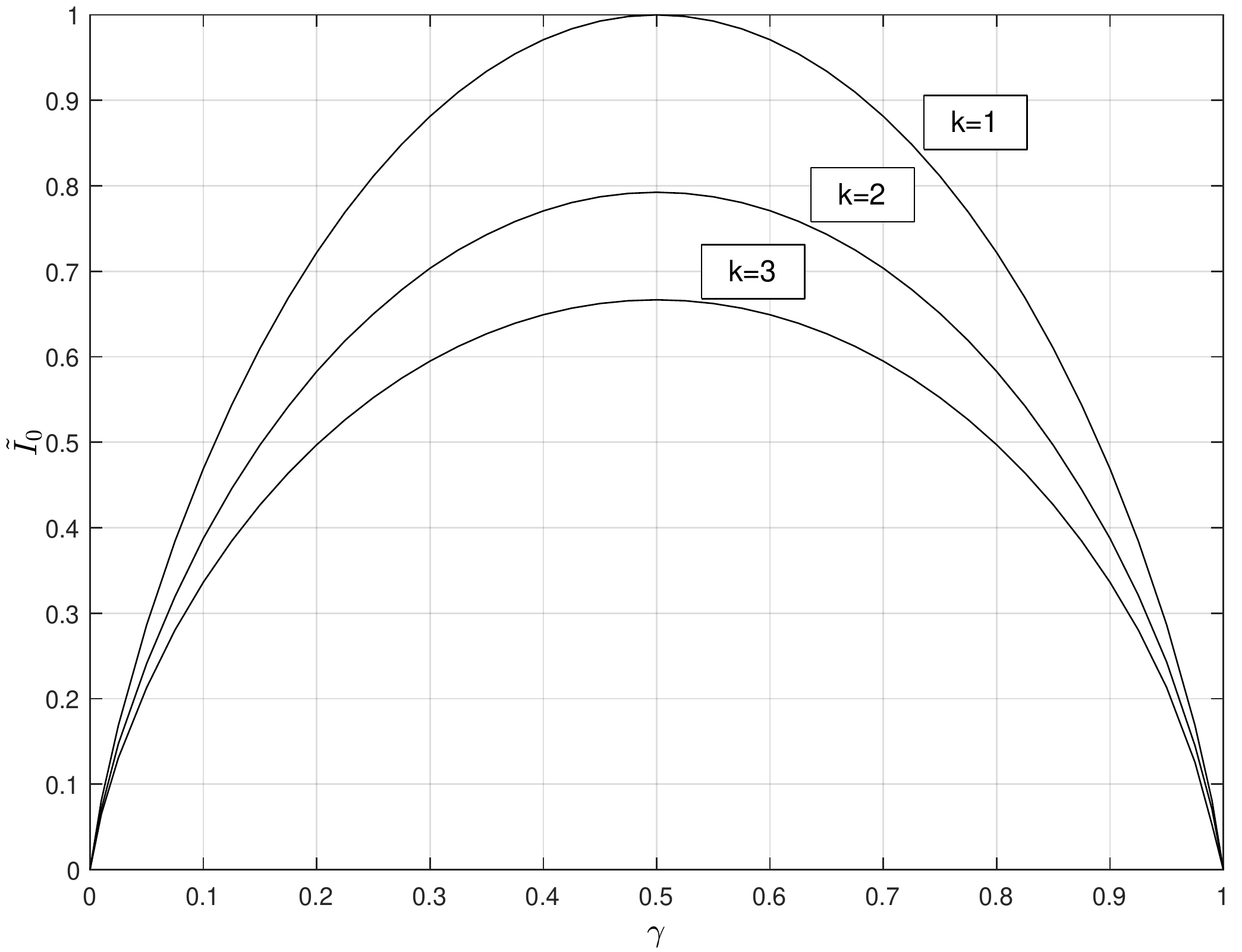}
\end{subfigure}%
\begin{subfigure}{.33\textwidth}
  \centering
  \includegraphics[scale=0.28]{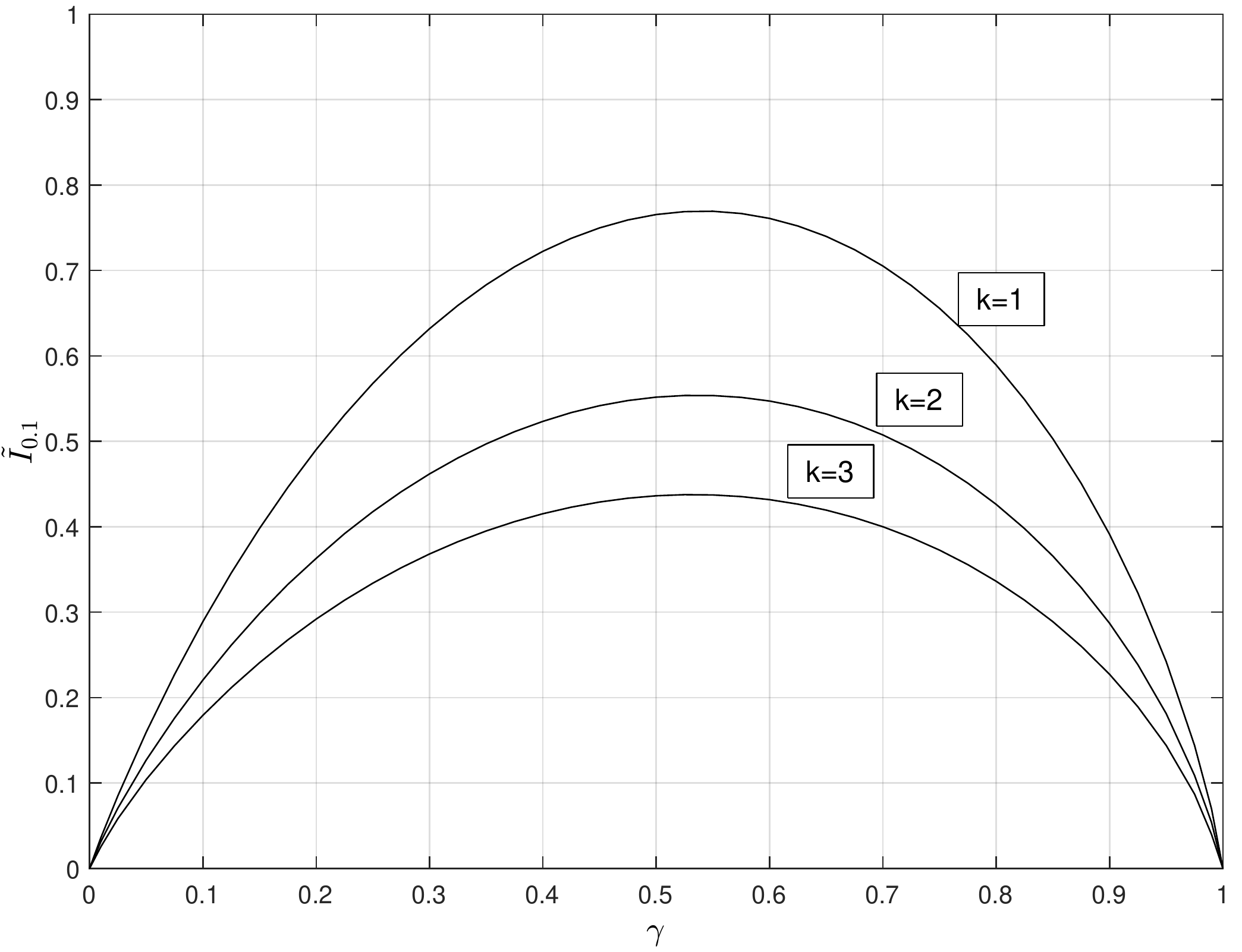}
\end{subfigure}
\begin{subfigure}{.33\textwidth}
  \centering
  \includegraphics[scale=0.28]{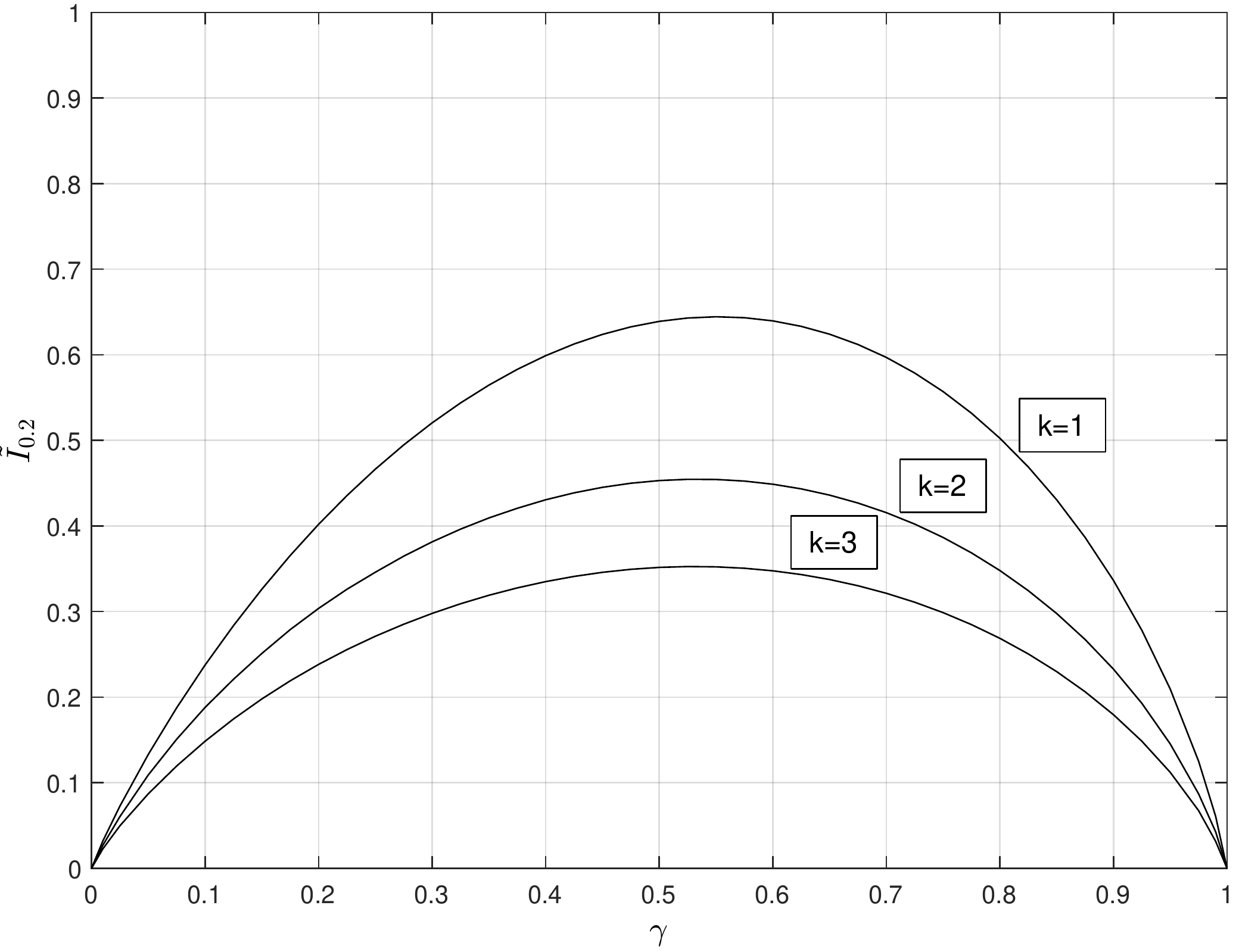}
\end{subfigure}
\caption{$\tilde{I}_{r_p}(\gamma,\frac{1}{k})$ for different values of $\gamma$ and $k\in\{1,2,3\}$ and $r_p\in\{0,0.1,0.2\}$.}
\label{fig:Itilde}
\end{figure*}

\begin{lemma}
\label{3userlem1}
The function $\tilde{I}_{r_p}$ could be computed using the following expression.
\begin{equation}
\displaystyle\tilde{I}_{r_p}(\gamma,\frac{1}{k})=\frac{1}{k}\check{H}_{r_p}(\gamma,\frac{1}{k})-\frac{1}{k}H(\textit{Bin}(k,r_p)),
\label{eq:I_rf}
\end{equation}
where $\displaystyle\check{H}_{r_p}(\gamma,\frac{1}{k})=\max_{\substack{X\in\{0,1,...,k\}\\ \mathbb{E}[X]=k\gamma}}H_{r_p}(Y)$ and the second term is the entropy of the binomial distribution with parameters $k$ and $r_p$.
\end{lemma}

\noindent
See Appendix \ref{app:E} for the proof of Lemma \ref{3userlem1}.

In order to calculate $\check{H}_{r_p}(\gamma,\frac{1}{k})$, the following optimization problem should be solved.
\begin{equation}
\begin{aligned}
&\displaystyle\max_{P_X\geq 0}\log_2e\sum_{i=0}^{2k}P_{Y}(i)\ln(\frac{1}{P_{Y}(i)})\\
&\hspace{15mm}\text{s.t.}\hspace{5mm} \left\{
\begin{array}{l l}
\sum_{i=0}^{k}iP_{X}(i)=\mathbb{E}[X]=k\gamma,\\
\sum_{i=0}^{k}P_{X}(i)=1,
\end{array} \right.
\end{aligned}
\label{eq:opti}
\end{equation}
where $P_Y=P_X*P_{\textit{Bin}(k,r_p)}$, that is,
\begin{equation}
\begin{aligned}
& P_{Y}(i)= \sum_{j=0}^{k}P_X(j)P_{\textit{Bin}(k,r_p)}(i-j),   \hspace{3mm}i\in\{0,1,...,2k\}.
\end{aligned}
\end{equation}

\noindent
Figure \ref{fig:Itilde} shows the functions $\tilde{I}_{0}(\gamma,\frac{1}{k})$, $\tilde{I}_{0.1}(\gamma,\frac{1}{k})$ and $\tilde{I}_{0.2}(\gamma,\frac{1}{k})$ for different values of $\gamma$ and $k\in\{1,2,3\}$.

\begin{lemma}
\label{3userlem2}
For all $0\le\ r_p\le1$, integers $1\le k_1, k_2, k_3\le\tau_{max}$, values $0\le\gamma_1,\gamma_2,\gamma_3\le1$, and $\alpha\in[0,1]$, such that $\displaystyle\alpha(\gamma_1,\frac{1}{k_1})+(1-\alpha)(\gamma_3,\frac{1}{k_3})=(\gamma_2,\frac{1}{k_2})$, we have
\begin{equation}
\alpha\tilde{I}_{r_p}(\gamma_1,\frac{1}{k_1})+(1-\alpha)\tilde{I}_{r_p}(\gamma_3,\frac{1}{k_3})\le\tilde{I}_{r_p}(\gamma_2,\frac{1}{k_2}).
\label{eq:toolem}
\end{equation}
\end{lemma}
\noindent
See Appendix \ref{app:F} for the proof of Lemma \ref{3userlem2}.

Using the mentioned properties, the capacity of the timing channel in the shared FCFS scheduler of Figure \ref{fig:3usersysmod} for different values of $r_p$ can be calculated. Figure \ref{fig:3userC} shows the value of the capacity with respect to $r_p$.

\begin{figure}[t]
\centering
\includegraphics[scale=0.32]{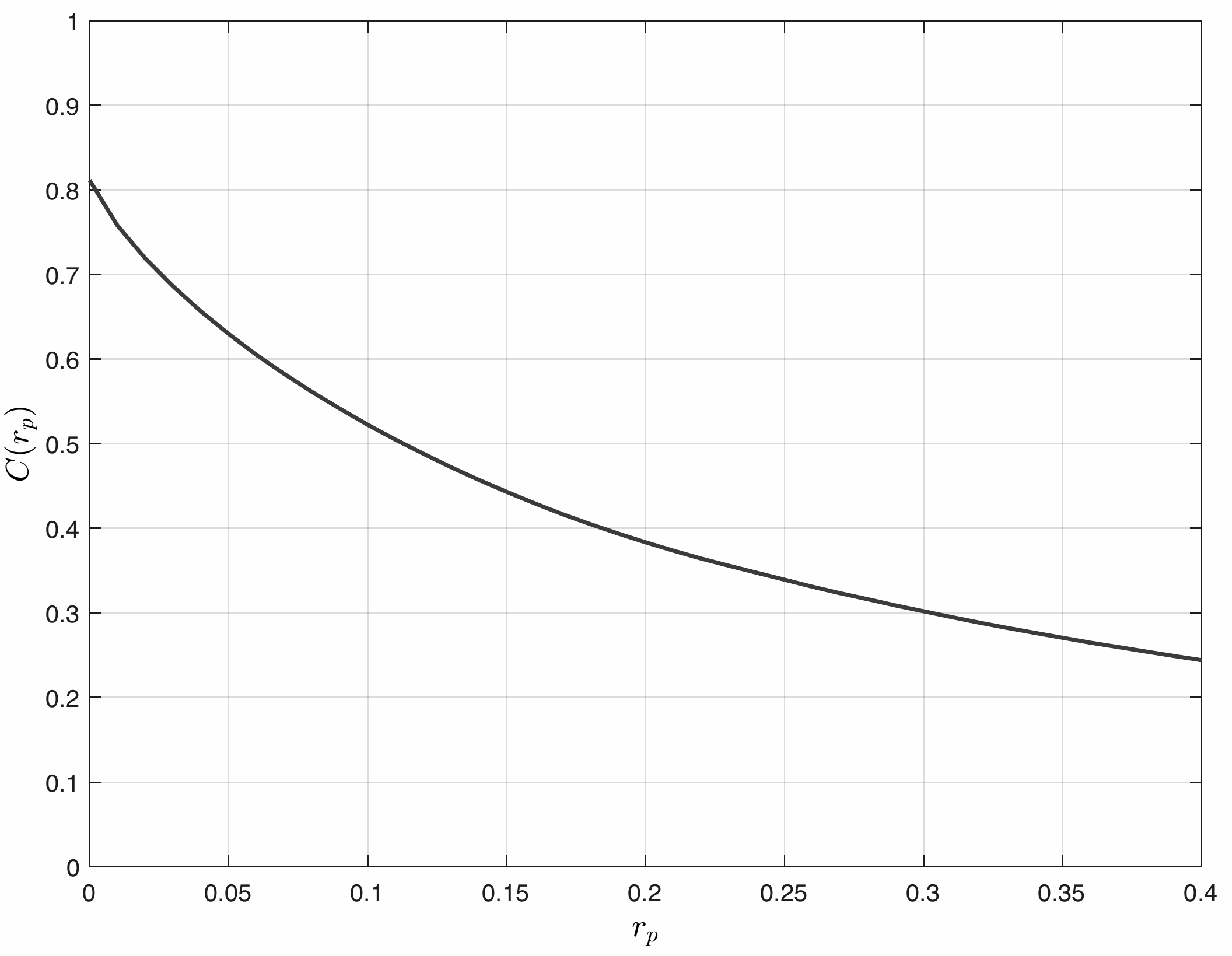}
\caption{capacity of the timing channel in the shared FCFS scheduler of Figure \ref{fig:3usersysmod} for different values of $r_p$.\vspace{-0mm}}
\label{fig:3userC}
\end{figure}

The following proof of Theorem \ref{3userthm} is based on converse and achievability arguments.
\subsection{Converse}
For the converse part, similar to the approach in Section \ref{2user}, first we find an upper bound on the information rate, which consists of a weighted summation of possible maximum information rates for different inter-arrival times of the packets in the probe stream, which satisfies the stability constraint; and then we further bound it with a summation which only corresponds to two inter-arrival times.

Suppose the rate of $U_p$ is $r_p$. Similar to the proof of Lemma \ref{2userlem4}, for the timing channel in a shared FCFS scheduler with service rate 1 depicted in Figure \ref{fig:3usersysmod}, any code consisting of a codebook of $M$ equiprobable binary codewords, where messages take on average $n$ time slots to be received, satisfies
\begin{equation*}
\begin{aligned}
\displaystyle\frac{1}{n}\log{M}&\le\frac{1}{n}I_{r_p}(W;\hat{W}|\tau^m)+\epsilon_n\\
&\overset{(a)}{\le}\frac{1}{n}I_{r_p}(W;Y^m|\tau^m)+\epsilon_n,
\end{aligned}
\end{equation*}
where $\epsilon_n=\frac{1}{n}(H(P_e)+P_e\log_2{(M-1)})$ and $(a)$ follows from data processing inequality in the model in Figure \ref{fig:3userprob}. Therefore,
\begin{equation*}
\begin{aligned}
\displaystyle\frac{1}{n}\log{M}&\leq\frac{1}{n}\sum_{j=1}^mI_{r_p}(W;Y_j|Y^{j-1}\tau^m)+\epsilon_n\\
&\leq\frac{1}{n}\sum_{j=1}^mI_{r_p}(W,Y^{j-1};Y_j|\tau^m)+\epsilon_n\\
&\overset{(a)}{\le}\frac{1}{n}\sum_{j=1}^mI_{r_p}(X_j;Y_j|\tau^m)+\epsilon_n\\
&\leq\frac{1}{n}\sum_{j=1}^m\max_{P_{X_j|\tau^m}}I_{r_p}(X_j;Y_j|\tau^m)+\epsilon_n,
\end{aligned}
\end{equation*}
where $(a)$ again follows from data processing inequality in the model in Figure \ref{fig:3userprob}. In the maximization above, the mean of the distribution $P_{X_j|\tau^m}$ is $\mathbb{E}[X_j|\tau^m]$ and in order to find the maximum information rate, the set of means, $\{\mathbb{E}[X_1|\tau^m],\mathbb{E}[X_2|\tau^m],...,\mathbb{E}[X_m|\tau^m]\}$, should satisfy the constraint $r_e+r_p+r_d= 1$, that is, $\frac{1}{n}\sum_{j=1}^m\mathbb{E}[X_j|\tau^m]+r_p+r_d= 1$.
Let $\xi_j=\frac{\mathbb{E}[X_j|\tau^m]}{\tau_j}$. Using \eqref{eq:i}, we have
\begin{equation}
\max_{\substack{P_{X_j|\tau^m}\\  \mathbb{E}[X_j|\tau^m]=\tau_j\xi_j}}I_{r_p}(X_j;Y_j|\tau^m)=\tau_j\tilde{I}_{r_p}(\xi_j, \frac{1}{\tau_j}).
\label{eq:3maxi1}
\end{equation}
This implies that
\begin{equation*}
\begin{aligned}
\displaystyle\frac{1}{n}\log{M}&\leq\frac{1}{n}\sum_{j=1}^m\tau_j\tilde{I}(\xi_j, \frac{1}{\tau_j})+\epsilon_n.
\end{aligned}
\end{equation*}

Next, using Lemma \ref{3userlem2} and similar to the proof of Lemma \ref{2userlem4}, by breaking the summation, using Jensen's inequality and the equation $n = \sum_{\tau=1}^{\tau_{max}} \tau\cdot m_{\tau}$, we will have
\begin{equation}
\begin{aligned}
\displaystyle\frac{1}{n}\log{M}&\le\sum_{\tau=1}^{\tau_{max}}\pi_{\tau}\tilde{I}_{r_p}(\mu_{\tau}, \frac{1}{\tau})+\epsilon_n,
\end{aligned}
\label{eq:generalsum3user}
\end{equation}
where $\mu_{\tau}$ is the average of $\xi_j$'s which have $\tau_j=\tau$ and $\pi_{\tau} = (\tau\cdot m_{\tau})/(\sum_{\tau=1}^{\tau_{max}} \tau\cdot m_{\tau})$.
In this expression, $\pi_{\tau}$ could be interpreted as the portion of time that user $U_d$ sends packets with inter-arrival time equal to $\tau$.

Also, using the same approach as the one used in the proof of Lemma \ref{2userlem4}, the constraint of the problem could be written as follows.

\begin{equation}
\displaystyle\sum_{\tau=1}^{\tau_{max}}\pi_{\tau}(\mu_{\tau}+\frac{1}{\tau}) = 1-r_p.
\label{eq:constrain23user}
\end{equation}
Suppose the set of pairs $\mathcal{S}=\{(\mu_{\tau},\frac{1}{\tau}),\tau\in[\tau_{max}]\}$ with weights $\{\pi_{\tau},\tau\in[\tau_{max}]\}$ gives $\sum_{\tau=1}^{\tau_{max}}{\pi}_{\tau}\tilde{I}_{r_p}({\mu}_{\tau}, \frac{1}{\tau})$ and has its operating point on the line $r_e+r_d=1-r_p$, and we have
$
\frac{1}{\tau^*}\leq\sum_{\tau=1}^{\tau_{max}}{\pi}_{\tau}\frac{1}{\tau}\leq\frac{1}{\tau^*+1}
$,
for some $1\le\tau^*\leq\tau_{max}-1$. We have
\begin{equation*}
\begin{aligned}
&\textstyle\beta_{\tau}(\mu_{\tau^*+1},\frac{1}{\tau^*+1})+(1-\beta_{\tau})(\mu_{\tau},\frac{1}{\tau})=(\mu^{\tau}_{\tau^*},\frac{1}{\tau^*}),~\tau\leq\tau^*-1,\\
&\textstyle\beta_{\tau}(\mu_{\tau^*},\frac{1}{\tau^*})+(1-\beta_{\tau})(\mu_{\tau},\frac{1}{\tau})=(\mu^{\tau}_{\tau^*+1},\frac{1}{\tau^*+1}),~\tau\geq\tau^*+2,
\end{aligned}
\end{equation*}
for some $\beta_{\tau}\in[0,1]$. Clearly, set
$
\{(\mu^1_{\tau^*},\frac{1}{\tau^*}),\cdots, (\mu^{\tau^*-1}_{\tau^*},\frac{1}{\tau^*}),$
$(\mu_{\tau^*},\frac{1}{\tau^*}),(\mu_{\tau^*+1},\frac{1}{\tau^*+1}), (\mu^{\tau^*+2}_{\tau^*+1},\frac{1}{\tau^*+1}), \cdots, (\mu^{\tau_{max}}_{\tau^*+1},\frac{1}{\tau^*+1})\}
$
can give the same operating point as $\mathcal{S}$ does. Therefore, using the technique presented in \eqref{eq:goodtech} and twice use of Lemma \ref{3userlem2} we have
\begin{equation*}
\begin{aligned}
\displaystyle \frac{1}{n}\log{M}&\leq\alpha\tilde{I}_{r_p}(\gamma_1,\frac{1}{\tau^*})+(1-\alpha)\tilde{I}_{r_p}(\gamma_2,\frac{1}{\tau^*+1})+\epsilon_n\\
\leq&\max_{\alpha , \gamma_1 , \gamma_2, \tau} \alpha\tilde{I}_{r_p}(\gamma_1,\frac{1}{\tau})+(1-\alpha)\tilde{I}_{r_p}(\gamma_2,\frac{1}{\tau+1})+\epsilon_n.
\end{aligned}
\end{equation*}
Letting $n\rightarrow\infty$, $\epsilon_n$ goes to zero and we get the desired result.
\subsection{Achievability}
For a given value of $r_p$, we first use the expression obtained in the converse to obtain optimal choice for parameters $\alpha$, $\gamma_1$, $\gamma_2$ and $\tau$. Because of Lemma \ref{3userlem2}, in order to find the optimal $\tau$, we can start with $\tau=1$ and optimize other parameters and then calculate $\alpha\tilde{I}_{r_p}(\gamma_1,\frac{1}{\tau})+(1-\alpha)\tilde{I}_{r_p}(\gamma_2,\frac{1}{\tau+1})$ and in each step, increase the value of $\tau$ by 1, stopping whenever the obtained value is decreased compared to the previous step. For instance, for $r_p\le0.1$, the optimal $\tau$ is 1, and hence the procedure stops after checking two steps. After calculating optimal parameters $\alpha^*$, $\gamma_1^*$, $\gamma_2^*$ and $\tau^*$, the optimal input distributions $P_1^*$ and $P_2^*$ could be obtained using optimization problem \eqref{eq:opti}.

Achieving the proposed upper bound could be done by a method exactly similar to the one used in Subsection \ref{ssec:ach}. The sequence of steps in our achievability scheme is as follows.

\noindent
\textbf{Generating the codebook:} 
The codebook $\mathcal{C}$ is generated by combining two codebooks $\mathcal{C}_1$ and $\mathcal{C}_2$, which are designed for the cases that the inter-arrival times in the probe stream are $\tau^*$ and $\tau^*+1$, respectively. 
\begin{itemize}
\item Set $\alpha =\alpha^*-\delta$, for a small and positive value of $\delta$.
\item To generate $\mathcal{C}_1$, consider the $(\tau^*+1)$-ary distribution $P^*_1$ over set of symbols $\{a_0, ..., a_{\tau^*}\}$. Generate a $(\tau^*+1)$-ary codebook $\mathcal{C}_1$ containing $2^{\alpha nR_1}$ sequences of length $\frac{1}{\tau^*}\alpha n$ of i.i.d. entries according to $P^*_1$. Substitute $a_i$ with a binary sequence of $i$ $1$'s followed by $\tau^*-i$ $0$'s. Therefore, we will have $2^{\alpha nR_1}$ binary sequences of length $\alpha n$.
\item To generate $\mathcal{C}_2$, consider the $(\tau^*+2)$-ary distribution $P^*_2$ over set of symbols $\{a_0, ..., a_{\tau^*+1}\}$. Generate a $(\tau^*+2)$-ary codebook $\mathcal{C}_2$ containing $2^{(1-\alpha) nR_2}$ sequences of length $\frac{1}{\tau^*+1}(1-\alpha) n$ of i.i.d. entries according to $P^*_2$. Substitute $a_i$ with a binary sequence of $i$ $1$'s followed by $\tau^*+1-i$ $0$'s. Therefore, we will have $2^{(1-\alpha)nR_2}$ binary sequences of length $(1-\alpha)n$.
\item Similar to the approach in Subsection \ref{ssec:ach}, combine $\mathcal{C}_1$ and $\mathcal{C}_2$ to get $\mathcal{C}$, such that $\mathcal{C}$ has $2^{n(\alpha R_1+(1-\alpha)R_2)}$ binary sequences of length $n$. Rows of $\mathcal{C}$ are our codewords.
\end{itemize}

The encoding and decoding parts are exactly similar to the approach in Subsection \ref{ssec:ach}, and we avoid repeating them, except that here $U_d^T$ sends  
a sequence with inter-arrival times of $\tau^*$ in the first $\alpha n$ time slots and $\tau^*+1$ for the rest of $(1-\alpha)n$ time slots.
Note that unlike the case with two users, here due to the presence of user $U_p$, we have noise in the channel. A standard random coding approach \cite[Chapter 7]{cover2012elements} shows that in infinite block-length regime, where $n\rightarrow\infty$, we can choose $R_1=\frac{1}{\tau^*}I_{r_p}(X;Y)$, with $X$ having distribution $P^*_1$, with channel described in \eqref{eq:yisbin}, i.e., $R_1=\tilde{I}_{r_p}(\gamma_1^*,\frac{1}{\tau^*})$, and $R_2=\frac{1}{\tau^*+1}I_{r_p}(X;Y)$, with $X$ having distribution $P^*_2$, with channel described in \eqref{eq:yisbin}, i.e., $R_2=\tilde{I}_{r_p}(\gamma_2^*,\frac{1}{\tau^*+1})$. 
Therefore, we have
\[
C\ge  \alpha\tilde{I}_{r_p}(\gamma^*_1,\frac{1}{\tau^*})+(1-\alpha)\tilde{I}_{r_p}(\gamma^*_2,\frac{1}{\tau^*+1}).
\]
Letting $\delta$ go to zero, completes the achievability proof.

\section{Conclusion}
\label{conclusion}

We studied convert queueing channels (CQCs) that can occur through delays experienced by users who are sharing a scheduler.
As the scheduling policy plays a crucial role in the possible information transmission rate in this type of channel, we focused on deterministic and work-conserving FCFS scheduling policy.
An information-theoretic framework was proposed to derive the capacity of the CQC under this scheduling policy.
First, we studied a system comprising a transmitter and a receiver that share a deterministic and work-conserving FCFS scheduler.
We obtained the maximum information transmission rate in this CQC and showed that an information leakage rate as high as 0.8114 bits per time slot is possible.
We also considered the effect of the presence of other users on the information transmission rate of this channel.
We extended the model to include a third user who also uses the shared resource and studied the effect of the presence of this user on the information transmission rate.
The solution approach presented in this extension can be applied to calculate the capacity of the covert queueing channel among any number of users.
The achievable information transmission rates obtained from this study demonstrate the possibility of significant information leakage and great privacy threats brought by CQCs in FCFS schedulers.
Based on this result, special attention must be paid to CQCs in high security systems.
Finding the capacity of CQCs under other scheduling policies, especially non-deterministic policies, remains to be done in the research area of covert communications and is considered as the main direction for future work.
Furthermore, a comprehensive study is required to design suitable scheduling policies that can simultaneously guarantee adequate levels of both security and throughput.


\begin{appendices}

\section{Proof of Stability}
\label{app:A}

Consider the system model with $M$ users and service rate $\rho$, and let $\mu=\sum_{i=1}^Mr_i$.
We denote arrival, service and queue length at time $k$, with $a(k)$, $s(k)$ and $q(k)$, respectively, and we have
\begin{equation*}
q(k+1)=(q(k)+a(k)-s(k))^+.
\end{equation*}
Using Foster-Lyapunov theorem with Lyapunov function $V(q(k))=(q(k))^2$ and calculating the drift, we have
\begin{equation*}
\begin{aligned}
\mathbb{E}[q^2(k+1)-q^2(k)|q(k)=q]&\le\mathbb{E}[(q+a-s)^2-q^2]\\
&=\mathbb{E}[2q(a-s)]+\mathbb{E}[(a-s)^2],
\end{aligned}
\end{equation*}
where $\mathbb{E}[(a-s)^2]$ is a constant and we denote it by $K$. Therefore, for some $\epsilon>0$, if $\mu<\rho$, for large enough value of $q$, we have
\begin{equation*}
\mathbb{E}[q^2(k+1)-q^2(k)|q(k)=q]\le2q(\mu-\rho)+K\le -\epsilon,
\end{equation*}
which implies the stability.

\section{Proof of Lemma \ref{2userlem0}}
\label{app:B}
In order to find the optimum distribution, $P_X$, the optimization problem could be written as follows.
\begin{equation}
\begin{aligned}
&\displaystyle\max_{P_X\geq 0}\log_2e\sum_{i=0}^{k}P_{X}(i)\ln(\frac{1}{P_{X}(i)})\\
&\hspace{15mm}s.t.\hspace{5mm} \left\{
\begin{array}{l l}
\sum_{i=0}^{k}iP_{X}(i)=\mathbb{E}[X]=k\gamma,\\
\sum_{i=0}^{k}P_{X}(i)=1,
\end{array} \right.
\end{aligned}
\label{eq:maxi}
\end{equation}
which could be solved using the Lagrange multipliers method. The Lagrangian function would be as follows.
\begin{equation*}
\begin{aligned}
\sum_{i=0}^{k}&P_{X}(i)\ln(\frac{1}{P_{X}(i)})+\lambda(\sum_{i=0}^{k}iP_{X}(i)-k\gamma)\\
&+\rho(\sum_{i=0}^{k}P_{X}(i)-1).
\end{aligned}
\end{equation*}
Setting the derivative with respect to $P_{X}(i)$ equal to zero, we get $\ln(\frac{1}{P_{X}(i)})-1+i\lambda+\rho=0$, which implies that
\begin{equation}
\label{eq:appII1}
\begin{aligned}
&P_{X}(i)=e^{\rho-1}\cdot e^{i\lambda}.
\end{aligned}
\end{equation}
Also, from the second constraint we have
\begin{equation}
\label{eq:appII2}
\begin{aligned}
& \sum_{i=0}^{k}e^{\rho-1}\cdot e^{i\lambda}=1\Rightarrow e^{\rho-1}=\frac{1}{\sum_{i=0}^{k}e^{i\lambda}}.
\end{aligned}
\end{equation}
Combining \eqref{eq:appII1} and \eqref{eq:appII2}, we have
\begin{equation*}
\begin{aligned}
P_{X}(i)=\frac{e^{i\lambda}}{\sum_{i=0}^{k}e^{i\lambda}},
\end{aligned}
\label{eq:dist}
\end{equation*}
which is the tilted distribution of $U_{k}$ with parameter $\lambda$.\\
In order to calculate $\lambda$, from the first constraint,
\begin{equation*}
\begin{aligned}
k\gamma&=\sum_{i=0}^{k}iP_{X}(i)
=\sum_{i=0}^{k}i\frac{e^{i\lambda}}{\sum_{i=0}^{k}e^{i\lambda}}
=\frac{\displaystyle\sum_{i=0}^{k}ie^{i\lambda}}{\displaystyle\sum_{i=0}^{k}e^{i\lambda}}\\
&=\frac{\mathbb{E}[U_{k}e^{U_{k}\lambda}]}{\mathbb{E}[e^{U_{k}\lambda}]}
=\frac{d}{d\lambda}(\ln{\mathbb{E}[e^{U_{k}\lambda}]})
=\psi'_{U_{k}}(\lambda).
\end{aligned}
\end{equation*}

\section{Proof of lemma \ref{2userlem2}}
\label{app:C}
First we note that
\begin{equation*}
\begin{aligned}
\tilde{H}&(\gamma,\frac{1}{k}) = \frac{1}{k}[\log_2(k+1)-\psi^*_{U_k}(k\gamma)\log_2 e]\\
&=[\frac{1}{k}\log_2(k+1)-\frac{1}{k}\underset{\lambda}{sup}\{k\gamma\lambda-\log(\frac{\sum_{i=0}^ke^{i\lambda}}{k+1})\}\log_2 e]\\
&=-\underset{\lambda}{sup}\{\gamma\lambda\log_2 e-\frac{1}{k}\log_2(\sum_{i=0}^ke^{i\lambda})\}.
\end{aligned}
\end{equation*}
Therefore, if we can show that for any given $\lambda$ the function $ h(\gamma,\frac{1}{k})=\gamma\lambda\log_2 e-\frac{1}{k}\log_2(\sum_{i=0}^ke^{i\lambda})$ is convex, then since the supremum of convex functions is convex, we can conclude the desired concavity of the function $\tilde{H}(\cdot,\cdot)$.\\
Noting that $ \frac{1}{k}\log_2(\sum_{i=0}^ke^{i\lambda})=\frac{1}{k}\log_2(\frac{1-e^{(k+1)\lambda}}{1-e^{\lambda}})$, to prove the convexity of $h(\cdot,\cdot)$, it suffices to prove that the function $ g(x)=x\log(\frac{1-e^{(\frac{1}{x}+1)\lambda}}{1-e^{\lambda}})$, $0<x\le1$, is concave. This is true from the concavity of the function $ \hat{g}(x)=\log(\frac{1-e^{(x+1)\lambda}}{1-e^{\lambda}})$, and the fact that for any function $f$, $xf(\frac{1}{x})$ is concave if $f(x)$ is concave. The concavity of the function $\hat{g}(\cdot)$ can be easily seen by taking its second derivative.

\section{Proof of Lemma \ref{2userlem3}}
\label{app:D}
For a given $\gamma$ and support set $\{0,1,...,k\}$, suppose the distribution $P_X^*$ is defined over $\{0,1,...,k\}$ and has mean $\mathbb{E}_{P^*}[X]=k\gamma$ and
\begin{equation*}
\max_{\substack{X\in\{0,1,...,k\}\\ \mathbb{E}[X]=k\gamma}} \frac{1}{k}H(X)=\frac{1}{k}H(P_X^*).
\end{equation*}
Define distribution $Q_X$ as 
$Q_X(i)=P_X^*(k-i)$, $0\le i\le k.$
Therefore the entropy of $Q_X$ will be the same as the entropy of $P_X^*$ and we have
\begin{equation*}
\begin{aligned}
\mathbb{E}_Q[X]&=\sum_{i=0}^kiQ_X(i)=\sum_{i=0}^kiP^*_X(k-i)\\
&=-\sum_{i=0}^k(-k+(k-i))P^*_X(k-i)\\
&=k-\sum_{i=0}^k(k-i)P^*_X(k-i)=k-k\gamma=k(1-\gamma).
\end{aligned}
\end{equation*}
Hence, we have
\begin{equation}
\begin{aligned}
\label{menav}
\tilde{H}(\gamma,\frac{1}{k})&=\max_{\substack{X\in\{0,1,...,k\}\\ \mathbb{E}[X]=k\gamma}} \frac{1}{k}H(X)=\frac{1}{k}H(P_X^*)=\frac{1}{k}H(Q_X)\\
&\le\max_{\substack{X\in\{0,1,...,k\}\\ \mathbb{E}[X]=k(1-\gamma)}} \frac{1}{k}H(X)=\tilde{H}(1-\gamma,\frac{1}{k}).
\end{aligned}
\end{equation}

Similarly, suppose for the distribution $Q_X^*$, defined over $\{0,1,...,k\}$ and with mean $\mathbb{E}_{Q^*}[X]=k(1-\gamma)$,
\begin{equation*}
\max_{\substack{X\in\{0,1,...,k\}\\ \mathbb{E}[X]=k(1-\gamma)}} \frac{1}{k}H(X)=\frac{1}{k}H(Q_X^*).
\end{equation*}
Define distribution $P_X$ as
$P_X(i)=Q_X^*(k-i)$, $0\le i\le k.$
Therefore the entropy of $P_X$ will be the same as the entropy of $Q_X^*$ and we have
\begin{equation*}
\begin{aligned}
\mathbb{E}_P[X]&=\sum_{i=0}^kiP_X(i)=\sum_{i=0}^kiQ^*_X(k-i)\\
&=-\sum_{i=0}^k(-k+(k-i))Q^*_X(k-i)\\
&=k-\sum_{i=0}^k(k-i)Q^*_X(k-i)=k-k(1-\gamma)=k\gamma.
\end{aligned}
\end{equation*}
Hence, we have
\begin{equation}
\begin{aligned}
\label{medov}
\tilde{H}(1-\gamma,\frac{1}{k})&=\max_{\substack{X\in\{0,1,...,k\}\\ \mathbb{E}[X]=k(1-\gamma)}} \frac{1}{k}H(X)=\frac{1}{k}H(Q_X^*)=\frac{1}{k}H(P_X)\\
&\le\max_{\substack{X\in\{0,1,...,k\}\\ \mathbb{E}[X]=k\gamma}} \frac{1}{k}H(X)=\tilde{H}(\gamma,\frac{1}{k}).
\end{aligned}
\end{equation}
Comparing (\ref{menav}) and (\ref{medov}) gives the desired result.

\section{Proof of Lemma \ref{3userlem1}}
\label{app:E}
\begin{equation*}
\begin{aligned}
\tilde{I}_{r_p}&(\gamma,\frac{1}{k})=\max_{\substack{X\in\{0,1,...,k\}\\ \mathbb{E}[X]=k\gamma}}\frac{1}{k}I_{r_p}(X;Y)\\
&=\max_{\substack{X\in\{0,1,...,k\}\\ \mathbb{E}[X]=k\gamma}} \frac{1}{k}[H_{r_p}(Y)-H_{r_p}(Y|X)]\\
&=\max_{\substack{X\in\{0,1,...,k\}\\ \mathbb{E}[X]=k\gamma}} \frac{1}{k}[H_{r_p}(Y)-\sum_{x=0}^kP_X(x)H_{r_p}(Y|X=x)]\\
&\overset{(a)}{=}\max_{\substack{X\in\{0,1,...,k\}\\ \mathbb{E}[X]=k\gamma}}\frac{1}{k}[H_{r_p}(Y)-\sum_{x=0}^kP_X(x)H(Bin(k,r_p))]\\
&=\max_{\substack{X\in\{0,1,...,k\}\\ \mathbb{E}[X]=k\gamma}} \frac{1}{k}[H_{r_p}(Y)-H(Bin(k,r_p))]\\
&=\max_{\substack{X\in\{0,1,...,k\}\\ \mathbb{E}[X]=k\gamma}} \frac{1}{k}H_{r_p}(Y)-\frac{1}{k}H(Bin(k,r_p))\\
&=\frac{1}{k}\check{H}_{r_p}(\gamma,\frac{1}{k})-\frac{1}{k}H(Bin(k,r_p)),
\end{aligned}
\end{equation*}
where $(a)$ follows from (\ref{eq:yisbin}).

\section{Proof of Lemma \ref{3userlem2}}
\label{app:F}
We first prove that the function $\tilde{I}(\cdot,\cdot)$ is concave in its first argument.
Let $P^*_{X_1}$ and $P^*_{X_3}$ be the optimum distributions resulted from optimization problem $\eqref{eq:opti}$ for parameters $(\gamma_1, \frac{1}{k})$ and $(\gamma_3, \frac{1}{k})$, respectively. Therefore for any $0\le\alpha\le1$,
\begin{equation*}
\begin{aligned}
&\alpha\tilde{I}_{r_p}(\gamma_1,\frac{1}{k})+(1-\alpha)\tilde{I}_{r_p}(\gamma_3,\frac{1}{k})\\
&\overset{(a)}{=}\frac{1}{k}\alpha H(P^*_{X_1}\ast Bin(k,r_p))+\frac{1}{k}(1-\alpha)H(P^*_{X_3}\ast Bin(k,r_p))\\
&~~~-\frac{1}{k}H(Bin(k,r_p))\\
&\overset{(b)}{\le}\frac{1}{k}H(\alpha(P^*_{X_1}\ast Bin(k,r_p))+(1-\alpha)(P^*_{X_3}\ast Bin(k,r_p)))\\
&~~~-\frac{1}{k}H(Bin(k,r_p))\\
&\le\frac{1}{k} \max_{\substack{X\in\{0,1,...,k\}\\ \mathbb{E}[X]=k(\alpha\gamma_1+(1-\alpha)\gamma_3)}}H(P_{X}\ast Bin(k,r_p))\\
&~~~-\frac{1}{k}H(Bin(k,r_p))\\
&=\tilde{I}_{r_p}(\alpha\gamma_1+(1-\alpha)\gamma_3,\frac{1}{k}),
\end{aligned}
\end{equation*}
where $(a)$ follows from Lemma \ref{3userlem1} and $(b)$ follows from the concavity of the entropy function.

\noindent
Because of the complexity and lack of symmetry or structure in the function $\tilde{I}$, there is no straightforward analytic method for proving its concavity. But we notice that it is suffices to show that for all $2\le k\le\tau_{max}-1$, and $\alpha$ such that
\begin{equation}
\label{eq:concproofconst}
\alpha\frac{1}{k-1}+(1-\alpha)\frac{1}{k+1}=\frac{1}{k},
\end{equation}
we have
\begin{equation}
\label{eq:concproof}
\begin{aligned}
\alpha\tilde{I}_{r_p}(\gamma_1,\frac{1}{k-1})+(1-&\alpha)\tilde{I}_{r_p}(\gamma_3,\frac{1}{k+1})\\
&\le\tilde{I}_{r_p}(\alpha\gamma_1+(1-\alpha)\gamma_3,\frac{1}{k}).
\end{aligned}
\end{equation}
From \eqref{eq:concproofconst} we have $\displaystyle\alpha=\frac{k-1}{2k}$, hence using Lemma \ref{3userlem1}, \eqref{eq:concproof} reduces to
\begin{equation}
\label{eq:concproofl}
2\check{H}_{r_p}(\gamma_2,\frac{1}{k})-\check{H}_{r_p}(\gamma_1,\frac{1}{k-1})-\check{H}_{r_p}(\gamma_3,\frac{1}{k+1})+f(k,r_p)\ge0,
\end{equation}
where $\displaystyle f(k,r_p)=H(Bin(k-1,r_p))+H(Bin(k+1,r_p))-2H(Bin(k,r_p))$. Noting that the left-hand side is a Lipschitz continuous function of $\gamma_1$, $\gamma_3$ (away from zero), and $r_p$ and the fact that $k$ takes finitely many values, the validation of inequality \eqref{eq:concproofl} can be done numerically.

\end{appendices}

\section*{Acknowledgment}
The authors would like to thank Mr. Rashid Tahir for helpful discussions regarding instances of existing covert queueing channels in real-world systems.\\
This work was in part supported by MURI grant ARMY W911NF-15-1-0479, Navy N00014-16-1-2804 and NSF CNS 17-18952.
\bibliographystyle{ieeetr}
\bibliography{thesisrefs}

\vspace{-2mm}

\begin{IEEEbiographynophoto}{AmirEmad Ghassami}
received the B.Sc. degree in electrical engineering from the Isfahan University of Technology in 2013 and the M.Sc. degree from the University of Illinois at Urbana–Champaign in 2015. He is currently a Graduate Student with the Department of Electrical and Computer Engineering, University of Illinois at Urbana–Champaign. His research interests include machine learning, causal structure learning, information theory, and statistical inference.
\end{IEEEbiographynophoto}

\vspace{-2mm}

\begin{IEEEbiographynophoto}{Negar Kiyavash} 
is a Willett Faculty Scholar and an Associate Professor, jointly, in the Departments of Industrial and Enterprise Systems Engineering and Electrical and Computer Engineering at the University of Illinois at Urbana-Champaign. She received the B.S. degree in Electrical and Computer Engineering from the Sharif University of Technology, Tehran, in 1999, and the M.S. and Ph.D. degrees, also in Electrical and Computer Engineering, both from the University of Illinois at Urbana-Champaign in 2003 and 2006, respectively. Her research interests include statistical signal processing, information theory, coding theory, operations research, and their application to learning and inference in complex networks. She is a recipient of the University of Illinois College of Engineerings Grainger Award in Emerging Technologies, AFOSR Young Investigator Award, NSF Career Award, and Deans Award for Excellence in Research at the University of Illinois at Urbana-Champaign.
\end{IEEEbiographynophoto}
\vfill

\end{document}